\documentclass[11pt,a4paper]{article}

\usepackage[margin=25mm]{geometry}
\usepackage{iftex}

\ifPDFTeX
  \usepackage[T1]{fontenc}
\fi

\usepackage{color}
\newcommand{\red}{}
\newcommand{\SY}{}

\usepackage{amsmath}
\usepackage{amsthm}
\usepackage{amsfonts}
\usepackage{authblk}
\usepackage[numbers,square]{natbib}
\usepackage{ulem}
\usepackage{comment}
\usepackage{tikz}
\usetikzlibrary{arrows.meta,positioning}
\usepackage[hidelinks]{hyperref}

\newcommand{\blk}[3]{\mathcal{B}_{#2,#3}\!\left(#1\right)}

\theoremstyle{plain}
\newtheorem{thm}{Theorem}
\newtheorem*{rthm1}{Theorem 1}
\newtheorem*{rthm2}{Theorem 2}
\newtheorem*{rthm3}{Theorem 3}
\newtheorem*{rthm4}{Theorem 4}
\newtheorem*{rthm5}{Theorem 5}
\newtheorem{lmm}[thm]{Lemma}

\newtheorem{dfn}[thm]{Definition}

\title{Multi-qubit controlled gate synthesis without T-count overhead in the small-error limit}
\author[1]{Soichiro Yamazaki\thanks{\href{mailto:soichiro.yamazaki@phys.s.u-tokyo.ac.jp}{\texttt{soichiro.yamazaki@phys.s.u-tokyo.ac.jp}}}}
\author[2,3]{Seiseki Akibue\thanks{\href{mailto:seiseki.akibue@ntt.com}{\texttt{seiseki.akibue@ntt.com}}}}
\author[2]{Kaoru Sano\thanks{\href{mailto:kaoru.sano@ntt.com}{\texttt{kaoru.sano@ntt.com}}}}
\affil[1]{The University of Tokyo, 7-3-1 Hongo, Bunkyo-ku, Tokyo 113-8654, Japan.}
\affil[2]{NTT Institute for Fundamental Mathematics, NTT Communication Science Labs., NTT, Inc., 3–1, Morinosato-Wakamiya, Atsugi, Kanagawa 243-0198, Japan}
\affil[3]{NTT Research Center for Theoretical Quantum Information, NTT, Inc., 3–1, Morinosato-Wakamiya, Atsugi, Kanagawa 243-0198, Japan.}

\date{}

\begin{document}

\maketitle

\begin{abstract}
We study the Clifford+T synthesis of quantum multiplexers $U=\bigoplus_{i=1}^{2^n}U_i$: an $n$-qubit control register selects which single-qubit gate $U_i\in\mathrm{SU}(2)$ acts on the target qubit.
For each block $U_i$, consider the minimum even T-count needed to approximate it within diamond distance $\varepsilon$ without ancillae, and let $m$ be the largest of these costs.
For an arbitrary multiplexer $U$, we construct an approximation of $U$ within the same distance using at most $m+O(\sqrt{2^n(m+1)})$ $T$ gates.
With one control qubit, the cost is $m+O(1)$ and no ancillae are required.
We also prove a lower bound that allows the Clifford+T circuit to use any number of clean ancillae initialized to zero and returned exactly to zero.
For fixed $n$ and independent Haar-random blocks, the minimum T-count among these circuits lies between $(3-\delta)\log_2(1/\varepsilon)$ and $(3+\delta)\log_2(1/\varepsilon)$ for every fixed $\delta>0$, except with probability $O(\varepsilon^{c_{n,\delta}})$ as $\varepsilon\rightarrow0$.
The optimal leading coefficient is therefore the same as for a single Haar-random $\mathrm{SU}(2)$ gate.
In the lower-bound proof, unitarity forces pairs of off-diagonal blocks describing transitions in opposite directions to satisfy cancellation relations with residuals of order $\varepsilon^2$.
We combine these relations with the arithmetic of Clifford+T circuits to bound the Haar measure of targets compatible with each pair of off-diagonal blocks.
As an application, we obtain ancilla-free approximations of Haar-random $\mathrm{SU}(4)$ gates within diamond distance $\varepsilon$, using at most $(9+\delta)\log_2(1/\varepsilon)$ $T$ gates for every fixed $\delta>0$, except with probability $O(\varepsilon^{c_{\delta}})$ with $c_{\delta}>0$.
\end{abstract}


\section{Introduction}
By the Gottesman-Knill theorem, Clifford circuits with inputs
and measurements in the computational basis can be simulated
efficiently on a classical computer.
Clifford gates together with the non-Clifford $T$ gate form a
universal gate set, allowing any quantum gate to be approximated
with arbitrary precision. In many fault-tolerant architectures,
implementing $T$ gates requires additional non-stabilizer resources,
such as distilled magic states, and is substantially more costly
than implementing Clifford operations~\cite{PhysRevA.71.022316,10.5555/3179330.3179331}.
This cost motivates minimizing the T-count, the number of $T$
gates in a circuit.
We call a circuit an $\varepsilon$-approximation of a target gate
when the diamond distance between their channels is at most
$\varepsilon$. The diamond distance quantifies how well the two
channels can be distinguished, including when an auxiliary
system is available.

For single-qubit Clifford+T synthesis, the worst-case minimum
T-count is $\Omega(\log(1/\varepsilon))$
\cite{10.5555/3179330.3179331}. Individual targets can have much
smaller costs: an exactly synthesizable gate, such as the identity,
has bounded cost as $\varepsilon\to0$.
For a Haar-random target in $\mathrm{SU}(2)$, the optimal leading
coefficient of $\log_2(1/\varepsilon)$ is $3$
\cite{Morisaki:2025mwy}. We specify the probability convention
for this asymptotic statement in Section~\ref{ssec:asymp}.
For multi-qubit targets, a decomposition into elementary gates
gives an upper bound on the synthesis cost. A matching lower
bound must also rule out less costly circuits that do not follow
that decomposition.

Our target is a quantum multiplexer, also called a uniformly
controlled single-qubit gate. It has the form
$U=\bigoplus_{i=1}^{2^n}U_i$: the $n$-qubit control register
selects $U_i$ when its state is $|i-1\rangle$, and $U_i$ acts on one
target qubit. A gate that acts only at one control value is
the special case in which all other blocks are the identity.
Our main results concern determinant-one blocks
$U_i\in\mathrm{SU}(2)$.
Controlled operations appear in state preparation
\cite{GossetKothariWu2026}, phase estimation
\cite{Nielsen_Chuang_2010}, quantum signal processing
\cite{PhysRevLett.118.010501}, and Hamiltonian simulation
\cite{Low2019hamiltonian}. Our deterministic construction
applies to arbitrary blocks, including those arising in these
applications. The probabilistic optimality statements use a
specific model: the blocks are drawn independently from Haar
measure on $\mathrm{SU}(2)$.

Gosset, Kothari, and Wu established worst-case T-count bounds
with optimal asymptotic order for state preparation and diagonal
synthesis. Their construction uses Boolean lookup to load
descriptions of single-qubit circuits~\cite{GossetKothariWu2026}.
The same lookup method gives an $O(\sqrt{2^n\log_2(1/\varepsilon)}+\log_2(1/\varepsilon))$ upper bound
for controlled $\mathrm{SU}(2)$ blocks.
For fixed $n$, this order bound leaves the coefficient of $\log_2(1/\varepsilon)$
undetermined. We show that its optimal value is $3$ for
independent Haar-random blocks. The lower bound allows the
approximating unitary to mix different control values and the
circuit to use any number of clean ancillae, initialized to zero
and returned exactly to zero.

\begin{thm}[Controlled $\mathrm{SU}(2)$ synthesis at fixed $n$]
\label{thm:su2cnt}
    Fix $n\geq1$, and let $U_1,\ldots,U_{2^n}$ be independent
    Haar-random elements of $\mathrm{SU}(2)$.
    Set $U=\bigoplus_iU_i$.
    Let $\mathcal T^{\mathrm{clean}}_\varepsilon(U)$ be the minimum
    T-count among Clifford+T circuits that approximate $U$
    within diamond distance $\varepsilon$, allowing any number
    of ancillae initialized to zero and returned exactly to zero.
    For every $\delta>0$, there are constants
    $C_{n,\delta},c_{n,\delta}>0$ such that, for sufficiently
    small $\varepsilon$,
    \begin{equation}\label{eq:thm1}
        \Pr\!\left[
          (3-\delta)\log_2(1/\varepsilon)\leq\mathcal T^{\mathrm{clean}}_\varepsilon(U)
                     \leq(3+\delta)\log_2(1/\varepsilon)
        \right]
        \geq1-C_{n,\delta}\varepsilon^{c_{n,\delta}}.
    \end{equation}
    The unitary induced on the control and target qubits need
    not preserve the block structure of $U$.
    Equation~\eqref{eq:thm1_instance} gives a construction bound
    for arbitrary blocks in terms of their individual synthesis
    costs.
\end{thm}

An approximating unitary can transfer amplitude between different
values of the control register. Its off-diagonal $2\times2$
blocks describe these transitions, which we call leakage.
Deleting these blocks need not produce a valid Clifford+T
unitary, so the lower-bound proof must account for them.
We use unitarity and the arithmetic of Clifford+T matrices to
restrict the possible pairs of leakage blocks, then bound the
Haar measure of targets compatible with each pair.
The arithmetic constraints depend on the number of control and
target qubits and on the T-count, but not on the number of
clean ancillae.

\begin{table}[t]
\centering
\small
\begin{tabular}{ccc}
\hline
 & Gosset--Kothari--Wu~\cite{GossetKothariWu2026} & This work \\
\hline
Upper bound & $O(\sqrt{2^n\log(1/\varepsilon)}+\log(1/\varepsilon))$ & $3\log_2(1/\varepsilon)+O(\sqrt{2^n\log(1/\varepsilon)})+o(\log(1/\varepsilon))$ \\[3pt]
Upper-bound model & all targets & Except with failure probability $O(\varepsilon^c)$ \\[3pt]
Lower bound & $\Omega(\sqrt{2^n\log(1/\varepsilon)}+\log(1/\varepsilon))$ & $3\log_2(1/\varepsilon)-o(\log(1/\varepsilon))\quad$ (for fixed $n$) \\[3pt]
Lower-bound model & Existence of such a target & \begin{tabular}{c}For fixed $n$,\\ except with failure probability $O(\varepsilon^{c_n})$ \end{tabular} \\
\hline
\end{tabular}
\caption{Comparison of this work and the previous work done by Gosset et al.\cite{GossetKothariWu2026}.
$c$ is a constant independent of $\varepsilon$ and $n$, while $c_n$ is also a constant independent only of $\varepsilon$.}
\label{tab:comparison}
\end{table}

With one control qubit, we can achieve the same leading coefficient
without ancillae. For arbitrary $U_1,U_2\in\mathrm{SU}(2)$, the
construction uses at most the larger of their individual minimum
even T-counts plus a constant; see
Equation~\eqref{eq:thm2_instance}. These individual costs are
measured at the same error tolerance $\varepsilon$ and without
ancillae.

\begin{samepage}
\begin{thm} [$\mathrm{SU}(2)\oplus\mathrm{SU}(2)$ ancilla-free controlled gate synthesis (informal)] \label{thm:su2su2}
    For two independent Haar-random $\mathrm{SU}(2)$ blocks,
    an ancilla-free Clifford+T circuit approximates their direct
    sum within diamond distance $\varepsilon$ using at most
    \begin{equation} \label{eq:thm2}
        3\log_2(1/\varepsilon) + o\left(\log(1/\varepsilon)\right)
    \end{equation}
    T gates with high probability as $\varepsilon\to0$.
    The leading coefficient $3$ is optimal even when any number
    of clean ancillae are allowed.
\end{thm}
\end{samepage}

\subsection{Applications}
We first extend Theorem~\ref{thm:su2cnt} to blocks in
$\mathrm{U}(2)$. The determinant phases of these blocks contribute
relative phases between control values, so they must be handled
along with the determinant-one factors. After removing an
overall phase, we decompose the target into two $\mathrm{SU}(2)$
multiplexers separated by permutations and apply the construction
of Theorem~\ref{thm:su2cnt} to both factors.


\begin{thm}[Controlled gate synthesis (informal)]
\label{thm:general_control}
    Fix $n\geq1$, and let
    \(
        U=\bigoplus_{j=1}^{2^n}U_j,
    \)
    where the $U_j$ are independent Haar-random elements
    of $\mathrm{U}(2)$.
    There is a Clifford+T construction with clean ancillae
    that approximates $U$ within diamond distance
    $\varepsilon$ and returns all ancillae exactly to zero.

    Unconditionally, for every fixed $\delta>0$, there are
    constants $C_{n,\delta},c_{n,\delta}>0$ such that
    \begin{equation}\label{eq:thm3}
        \Pr\!\left[
            \mathcal T^{\mathrm{clean}}_\varepsilon(U)
            \leq (8+\delta)\log_2(1/\varepsilon)
        \right]
        \geq
        1-C_{n,\delta}\varepsilon^{c_{n,\delta}}
    \end{equation}
    for all sufficiently small $\varepsilon>0$.

    Assuming the generic-angle part of the
    Ross--Selinger conjecture
    \cite[Section~8.3, Conjecture~1]{10.5555/3179330.3179331},
    the leading coefficient improves from $8$ to $6$:
    for every fixed $\delta>0$,
    \begin{equation}\label{eq:thm3_conjectural}
        \Pr\!\left[
            \mathcal T^{\mathrm{clean}}_\varepsilon(U)
            \leq (6+\delta)\log_2(1/\varepsilon)
        \right]
        \longrightarrow 1
        \qquad (\varepsilon\to0).
    \end{equation}
    Under the stronger quantitative assumption in
    Equation~\eqref{eq:rz_quantitative_hypothesis},
    this conditional bound also holds with failure
    probability $O(\varepsilon^{c_{n,\delta}})$
    for some $c_{n,\delta}>0$.
\end{thm}

Repeated use of the two-block construction in
Theorem~\ref{thm:su2su2} also gives circuits without ancillae
when there are multiple control qubits. In this construction,
the leading T-count coefficient increases with $n$.

\begin{thm}[Ancilla-free controlled gate synthesis (informal)]
    \label{thm:su2noancilla}
    For fixed $n$ and independent Haar-random blocks
    $U_i\in\mathrm{SU}(2)$, an ancilla-free Clifford+T circuit
    approximates $\bigoplus_iU_i$ within diamond distance
    $\varepsilon$ using at most
    \begin{equation}\label{eq:thm4}
        \frac{3(2^n+1)}2\log_2(1/\varepsilon)
        +O(n2^n)+o\!\left(2^n\log(1/\varepsilon)\right)
        \quad\mathrm{w.h.p.}
    \end{equation}
    T gates. The successive single-qubit targets depend on
    previously chosen approximate circuits. We bound their
    joint density relative to product Haar measure uniformly
    over these choices, which permits the use of the
    single-qubit T-count bounds.
\end{thm}

The two-block construction also gives ancilla-free approximations
of general two-qubit gates.

\begin{thm} [SU(4) gate synthesis (informal)] \label{thm:su4synthesis}
    For a Haar-random target in $\mathrm{SU}(4)$, an ancilla-free
    Clifford+T circuit achieves diamond distance at most
    $\varepsilon$ using at most
    \begin{equation} \label{eq:thm5}
        9\log_2(1/\varepsilon) + o(\log(1/\varepsilon))
    \end{equation}
    T gates with high probability as $\varepsilon\to0$.
\end{thm}

The standard KAK-based circuit contains four $\mathrm{SU}(2)$
gates and three $z$-rotations up to Clifford conjugation
\cite{VatanWilliams2004}.
For Haar-random targets, the local factors and rotation angles
have bounded densities relative to Haar and uniform angular
measure, respectively~\cite{WattsOConnorVala2013}.
Thus, synthesizing these factors separately using
Lemmas~\ref{lmm:even_tcount_su2}
and~\ref{lmm:rotation_upper_tail} yields an upper bound
with leading T-count coefficient $3\times4+5\times3=27$
w.h.p.
Under the generic-angle part of the Ross--Selinger conjecture
\cite[Section~8.3, Conjecture~1]{10.5555/3179330.3179331},
this coefficient improves to $3\times4+3\times3=21$,
with success probability tending to one as $\varepsilon\to0$.
In contrast, Theorem~\ref{thm:su4synthesis} gives an
unconditional upper bound with leading coefficient $9$.
To justify the probability bound in Theorem~\ref{thm:su4synthesis}, we derive the distributions of the single-qubit targets produced by this construction in Appendix~\ref{sec:SU(4)synthesis}.

\subsection{Discussion}
For fixed $n$ and independent Haar-random $\mathrm{SU}(2)$ blocks,
Theorem~\ref{thm:su2cnt} identifies $3$ as the optimal coefficient
of $\log_2(1/\varepsilon)$. For each fixed $\delta>0$,
the probability that the minimum T-count differs from $3\log_2(1/\varepsilon)$
by more than $\delta \log_2(1/\varepsilon)$ is bounded by a positive power of
$\varepsilon$.
The lower bound permits arbitrarily many clean ancillae, so these
resources cannot reduce the leading coefficient.
The sharp dependence of the synthesis cost on $n$ remains open.

The clean-ancilla condition requires every ancilla to be
initialized to zero and returned exactly to zero.
Our lower bound therefore does not cover channel
approximations obtained by discarding ancillae that are not
returned clean, or protocols using measurements and classical
feed-forward.

The target family $\mathrm{SU}(2)^{\oplus 2^n}$ is a
block-diagonal subgroup of $\mathrm{U}(2^{n+1})$.
Each block acts on the two-dimensional subspace associated
with one control value. An approximating unitary may have
off-diagonal $2\times2$ blocks, whose entries are transition
amplitudes between these subspaces; we call these transitions
leakage. A lower bound must account for these approximants.
Discarding their off-diagonal blocks generally destroys
unitarity, while projecting onto a unitary in the subgroup
need not preserve the Clifford+T arithmetic or the T-count.

Approximation within diamond distance $\varepsilon$ bounds
the operator norm of each leakage block by $O(\varepsilon)$.
Unitarity also relates leakage in opposite directions. After aligning the global phase and
removing the target block rotations, an off-diagonal block
and the adjoint of the block for the reverse transition have
a sum whose operator norm is $O(\varepsilon^2)$, as expressed in
Equation~\eqref{eq:clean_block_cancellation}.
Thus the leakage blocks cannot vary independently.
We combine this constraint with the arithmetic constraints of
Lemma~\ref{lmm:clean_arithmetic}, the matrix-counting estimates of
Lemmas~\ref{lmm:clean_fixed_gram} and~\ref{lmm:clean_trace_determinant},
and the Haar small-ball bound of Lemma~\ref{lmm:clean_small_ball}.
The counting estimates bound the number of possible pairs of
blocks; the small-ball bound controls the Haar measure of target
blocks compatible with each pair. Together, they bound the
probability that a target admits a general approximating unitary
with a given T-count.

This argument suggests a route to resource lower bounds for
other structured compilation tasks. For another target
family, one would need to count low-T-count unitaries in a
small neighbourhood of that family, including approximants
that leave its invariant subspaces. Such an extension would
require arithmetic and measure estimates for the new family,
together with geometric constraints such as the quadratic
cancellation relation used here.

When $n$ also grows as $\varepsilon\to0$, the displayed
square-root contribution $\sqrt{2^n\log_2(1/\varepsilon)}$ is negligible relative
to $\log_2(1/\varepsilon)$ if $2^n=o(\log_2(1/\varepsilon))$. Any claim in this simultaneous limit
must also satisfy the explicit probability bounds in
Section~\ref{ssec:tcount_analysis}.
We do not determine the optimal coefficient of the
square-root term.

Our resource statements concern the number of $T$ gates in
the resulting circuit. They do not bound the classical
running time needed to find the single-qubit approximants:
selecting a minimum-length word, or proving that a suitable
word exists, does not provide a compiler running in time
polynomial in $\log(1/\varepsilon)$.

\section{Preliminaries}
\subsection{Basic notations}
A computational basis is denoted by $\{|0\rangle,|1\rangle,\cdots,|2^n-1\rangle\}$.
Each basis vector
\begin{equation*}
    |i\rangle=|i_ni_{n-1}\cdots i_1 {}_{(2)}\rangle
\end{equation*}
corresponds to the first qubit being $|i_1\rangle$, the second qubit being $|i_2\rangle$, $\cdots$, and the $n$-th qubit being $|i_n\rangle$ {\red, where the subscript ${}_{(2)}$ denotes the binary representation}.
An $n$-qubit quantum gate is a member of the $2^n\times 2^n$ unitary group
\begin{equation*}
    \mathrm{U}(2^n)=\{U\in \mathbb{C}^{2^n\times 2^n}|~U^\dagger U=UU^\dagger=I_{2^n}\},
\end{equation*}
where $I_{2^n}$ denotes the identity.
Since the global phase has no physical meaning, we can identify two gates that differ only by the global phase, and the gates are also considered as members of the special unitary group $\mathrm{SU}(2^n)=\{U\in \mathrm{U}(2^n)|~ \det(U)=1\}$.
\begin{dfn}[$n$-controlled gate] \label{dfn:n-controlled-gate}
    Let $\mathcal{A}$ be an $n$-qubit control register and $\mathcal{B}$ be a single-qubit target register.
    An $n$-controlled gate is an $(n+1)$-qubit unitary of the form
    \begin{equation}
        \begin{aligned}
            U
            &= \sum_{i=0}^{2^n-1}|i\rangle\langle i|_{\mathcal{A}}\otimes A_{i+1} \\
            &= \begin{pmatrix}
                A_1\\
                &A_2\\
                &&\ddots\\
                &&&A_{2^n}
            \end{pmatrix}
            = \bigoplus_{i=1}^{2^n}A_i,
        \end{aligned}
    \end{equation}
    where each $A_i\in\mathrm{U}(2)$ acts on $\mathcal{B}$.
    We denote the set of all such gates by $\mathrm{U}(2)^{\oplus 2^n}$.
    Its subset $\mathrm{SU}(2)^{\oplus 2^n}$ consists of the $n$-controlled gates for which $A_i\in\mathrm{SU}(2)$ for every $i$.
\end{dfn}

A quantum channel from $n$-qubit states to $n$-qubit states is completely positive and trace preserving (CPTP) mapping.
Given two quantum channels, the diamond distance is a metric that measures the maximum {\red bias of the success} probability of distinguishing between them when the input can be freely chosen.
\begin{dfn} [Diamond distance]
    Let $\mathcal{E}_1, \mathcal{E}_2$ be two quantum channels from $n$-qubit states to $n$-qubit states.
    \begin{equation}
        D_\diamond(\mathcal{E}_1, \mathcal{E}_2)=
        \frac{1}{2}\max_{\rho}\|((\mathcal{E}_1-\mathcal{E}_2)\otimes\mathcal{I}_n)(\rho) \|,
    \end{equation}
    where $\mathcal{I}_n$ is an identity channel on a $n$-qubit system and $\|\cdot\|$ is a trace norm.
\end{dfn}
A quantum channel $\mathcal{E}_U$ of unitary $U$ is $\mathcal{E}_U(\rho)=U\rho U^\dagger$ where $\rho$ is a positive semidefinite matrix with trace 1, or a density operator.
The diamond {\red distance} between two unitary channels $U, V$ can be calculated by the following equation:
\begin{equation}
    D_\diamond(\mathcal{E}_U, \mathcal{E}_V) = \sqrt{1-\min_\rho\left|\mathrm{tr}\left(\rho UV^\dagger\right)\right|^2}.
\end{equation}
For single qubit unitary channels, this equation is reduced into
\begin{equation}
    D_\diamond(\mathcal{E}_U, \mathcal{E}_V) = \sqrt{1-\left|\frac{1}{2}\mathrm{tr}\left( UV^\dagger\right)\right|^2}.
\end{equation}
{\red Throughout the paper, we adopt a slight abuse of notation and use $D_\diamond(U,V)$ to represent $D_\diamond(\mathcal{E}_U, \mathcal{E}_V)$.}

Channel representation $\widehat{U}$ of an $n$-qubit unitary $U\in\mathrm{U}(2^n)$ can be explicitly written as a member of $\mathrm{O}(4^n)$:
\begin{equation}
    \widehat{U}_{i,j} = \frac{1}{2^n}\mathrm{tr}(P_iUP_jU^\dagger),
\end{equation}
where $i,j=0,1,\cdots,4^n-1$ and
\begin{equation} \label{eq:pauli_time}
    P_i=\bigotimes_{k=1}^{n}\sigma_{i_k}
    \qquad
    (\sigma_0=I_2,~\sigma_1=X,~\sigma_2=Y,~\sigma_3=Z,~(i_1i_2\cdots i_n)_{(4)}=i).
\end{equation}
Equivalently, $\widehat{U}_{i,j}$ is the coefficient of $P_i$ in the Pauli expansion of $UP_jU^\dagger$:
\begin{equation*}
    UP_jU^\dagger=\sum_{i=0}^{4^n-1}\widehat{U}_{i,j}P_i.
\end{equation*}

Note that the following relations hold for the channel representation:
\begin{enumerate}
    \item $\widehat{UV}=\widehat{U}\widehat{V}$
    \item $\widehat{I}=I$
    \item $\widehat{U\otimes V}=\widehat{U}\otimes\widehat{V}$
\end{enumerate}

\subsection{Asymptotic notation} \label{ssec:asymp}
All probability statements specify the distribution of the target
unitary. For fixed values of the other parameters, an event
$E_\varepsilon$ holds with high probability (w.h.p.) as
$\varepsilon\to0$ if there are constants $C,c>0$ such that
\begin{equation}
    \Pr[E_\varepsilon]\geq1-C\varepsilon^c
\end{equation}
for all sufficiently small $\varepsilon$. These constants and the
small-error threshold may depend on the parameters held fixed.
For a random variable $X_\varepsilon$, the notation
$X_\varepsilon=a\log_2(1/\varepsilon)
+o(\log(1/\varepsilon))$ w.h.p. means that, for every $\delta>0$,
\begin{equation}
    \Pr\!\left[
        \left|X_\varepsilon-a\log_2(1/\varepsilon)\right|
        \leq\delta\log_2(1/\varepsilon)
    \right]
    \geq1-C_\delta\varepsilon^{c_\delta}
\end{equation}
for suitable $C_\delta,c_\delta>0$ and sufficiently small
$\varepsilon$. When the number of qubits also varies, we give an
explicit failure-probability bound and state which constants are
uniform. In particular, an upper bound containing an $O(\cdot)$
term does not assert concentration around a fixed coefficient of
that term.

\subsection{Clifford+T gate synthesis}
A set of single-qubit Pauli gates is defined as follows:
\begin{equation}
    \mathrm{Pauli}_1 = \left\{
    i^k\sigma\middle|~
    k\in\mathbb{Z},~
    \sigma\in\left\{
    I_2=\begin{pmatrix}
        1&0\\
        0&1
    \end{pmatrix},\,
    X=\begin{pmatrix}
        0&1\\
        1&0
    \end{pmatrix},\,
    Y=\begin{pmatrix}
        0&-i\\
        i&0
    \end{pmatrix},\,
    Z=\begin{pmatrix}
        1&0\\
        0&-1
    \end{pmatrix}
    \right\}
    \right\},
\end{equation}
and a set of $n$-qubit Pauli gates is defined as follows:
\begin{equation}
    \mathrm{Pauli}_n = \left\{
    \bigotimes_{i=1}^n \sigma_i
    \middle|~
    \sigma_i\in\mathrm{Pauli}_1
    \right\}.
\end{equation}
When a group operation is defined as the matrix product, it forms the $n$-qubit Pauli group, and its normalizer is the $n$-qubit Clifford group written as $\mathrm{Clifford}_n$.
Ignoring the global phase, all elements of $\mathrm{Clifford}_n$ can be identified with elements that can be constructed using a finite number of $H$, $S$, and CNOT gates defined as follows (for $n=1$, a CNOT gate is not needed):
\begin{equation}
    H= \frac{1}{\sqrt{2}}\begin{pmatrix}
        1&1\\
        1&-1
    \end{pmatrix},~
    S=\begin{pmatrix}
        1&0\\
        0&i
    \end{pmatrix},~
    \mathrm{CNOT}=\begin{pmatrix}
        1&0&0&0\\
        0&1&0&0\\
        0&0&0&1\\
        0&0&1&0
    \end{pmatrix}.
\end{equation}
Adding the $T$ gate defined below to the Clifford gate set yields Clifford+T, which allows approximation of any unitary with arbitrary precision:
\begin{equation}
    T=\begin{pmatrix}
        1&0\\
        0&\zeta
    \end{pmatrix},
\end{equation}
where $\zeta=e^{i\pi/4}$.
When using matrix representatives, we choose Cliffords from
the group generated by $H$, $S$, and CNOT (including the
scalar $\zeta I$). In particular, single-qubit representatives
form a finite group and have fourth-root determinants.
Arbitrary global phases are identified only when comparing
unitary channels, not when controlling an individual block.
Given a target $n$-qubit unitary $U$ and $a$ ancillary qubits,
let $\mathcal{P}_a(\rho)=\rho\otimes|0^a\rangle\langle0^a|$
be the ancilla-preparation channel. A Clifford+T circuit $W$
on $n+a$ qubits induces the logical channel
\begin{equation}
    \Phi_W=\operatorname{tr}_a\circ\mathcal{E}_W\circ\mathcal{P}_a.
\end{equation}
Approximate synthesis requires
$D_\diamond(\mathcal{E}_U,\Phi_W)\leq\varepsilon$.
We allow equality at the error tolerance throughout the synthesis
definitions: a circuit whose error is exactly $\varepsilon$ is admissible.
This definition permits a general output channel. In the exact
clean-return model used for our main lower bound, one additionally
requires $W(|\psi\rangle|0^a\rangle)=V|\psi\rangle|0^a\rangle$
for every $|\psi\rangle$ and a logical unitary $V$.
The T-count is the number of T gates in the circuit; Clifford
gates and the number of ancillary qubits are separate resources.

Some controlled gates can be exactly implemented by Clifford+T.
Specifically, the controlled $H$ gate $I_2\oplus H$ can be exactly implemented by a T-count 2 Clifford+T circuit, and the controlled $S$ gate $I_2\oplus S$ can be exactly implemented by a T-count 3 Clifford+T circuit \cite{Amy2012AMA}.

The Matsumoto--Amano normal form gives a T-optimal exact representation of every single-qubit Clifford+T operator~\cite{GilesSelinger2013NormalForm}.
Regrouping its syllables yields
\begin{equation}
\label{eq:manormal}
    F\left(\prod_{j=1}^{m}M_j\right)C,
    \qquad
    F\in\{I_2,H,SH\},\quad
    M_j\in\{TH,TSH\},\quad C\in\mathrm{Clifford}_1.
\end{equation}
Here $m$ is the minimum T-count for the represented operator.
We use only the existence of this representation: with the
unrestricted final Clifford $C$, the regrouped expression in
Equation~\eqref{eq:manormal} need not be unique. The parity of
the T-count is preserved by normalization, since the determinant
of a Clifford is a fourth root of unity and $\det T=\zeta$.

\section{\texorpdfstring{$n$-controlled gate synthesis}{n-controlled gate synthesis}}
We prove Theorem~\ref{thm:su2cnt} in the following precise form.
For $U\in\mathrm{SU}(2)$, let $\mathcal{T}_\varepsilon(U)$ denote
the minimum even T-count of an ancilla-free Clifford+T circuit
$V$ satisfying $D_\diamond(U,V)\leq\varepsilon$.
For a controlled target $U$, write
$\mathcal{T}^{\mathrm{con}}_\varepsilon(U)$ for the T-count of
the circuit constructed in Section~\ref{ssec:construction}.
This construction cost is distinct from the minimum
$\mathcal{T}^{\mathrm{clean}}_\varepsilon(U)$ over all circuits
with arbitrarily many exactly restored clean ancillae, defined
in Section~\ref{ssec:lower_bound}.
\begin{rthm1}[$\mathrm{SU}(2)^{\oplus 2^n}$ controlled gate synthesis]
    Let $n\geq1$ and
    $U=\bigoplus_{i=1}^{2^n}U_i$ with $U_i\in\mathrm{SU}(2)$.
    Set $m=\max_i\mathcal{T}_\varepsilon(U_i)$. There is an
    absolute constant $C>0$ such that the construction satisfies
    \begin{equation} \label{eq:thm1_instance}
        \mathcal{T}^{\mathrm{con}}_\varepsilon(U)
        \leq m+C\sqrt{2^n(m+1)}.
    \end{equation}
    It uses $O(m+n+1)$ clean ancillae and
    $O(\sqrt{2^n(m+1)})$ dirty ancillae, all restored exactly.

    Suppose now that the $U_i$ are independent Haar-random
    elements of $\mathrm{SU}(2)$. For any $\delta>0$ and $\iota>1$,
    define
    \begin{equation}
        s=2\left\lceil
            \frac{(3+\delta)\log_2(1/\varepsilon)+\iota n}{2}
        \right\rceil.
    \end{equation}
    There is an absolute constant $C_0,C_{\delta,\iota}>0$ such that, uniformly for $n\geq1$ and $0<\varepsilon\leq1/2$,
    \begin{equation}\label{eq:thm1_full}
        \Pr\!\left[
            \mathcal{T}^{\mathrm{con}}_\varepsilon(U)
            >s+C_0\sqrt{2^n\log_2(1/\varepsilon)}
        \right]
        \leq
        C_{\delta,\iota}\,
        2^{n(1-\iota)}\varepsilon^\delta
        \bigl(\log_2(1/\varepsilon)+n+1\bigr)^2.
    \end{equation}
    In particular, for every fixed $n\geq1$ and $\delta>0$,
    there are $C_{n,\delta},c_{n,\delta}>0$ such that
    \begin{equation} \label{eq:thm1_fixed_n}
        \Pr\!\left[
            (3-\delta)\log_2(1/\varepsilon)
            \leq\mathcal{T}^{\mathrm{clean}}_\varepsilon(U)
            \leq(3+\delta)\log_2(1/\varepsilon)
        \right]
        \geq1-C_{n,\delta}\varepsilon^{c_{n,\delta}}
    \end{equation}
    for sufficiently small $\varepsilon$.
\end{rthm1}
The distribution in the probabilistic statements is the product
Haar measure on the $2^n$ target blocks. No such distribution is
assumed in Equation~\eqref{eq:thm1_instance}. In particular,
that deterministic bound also applies when blocks are repeated
or fixed to the identity.
Equation~\eqref{eq:thm1_full} retains the dependence on both
$n$ and $\varepsilon$; its square-root term has size
$O(\sqrt{2^n\log_2(1/\varepsilon)})$. The construction has leading T-count
coefficient at most $3$ in the regime $2^n=o(\log_2(1/\varepsilon))$, while the
matching lower bound in Equation~\eqref{eq:thm1_fixed_n} is
asserted for fixed $n$.

Section~\ref{ssec:construction} describes the circuit
construction. Section~\ref{ssec:lower_bound} establishes a
lower bound for arbitrary approximations with clean ancillae.
Section~\ref{ssec:tcount_analysis} combines the resource count
and the probability estimates to prove the theorem.
Section~\ref{ssec:thm2} gives an ancilla-free construction for
the two-qubit case.

\subsection{Circuit construction} \label{ssec:construction}
Given an $n$-controlled gate $U_1\oplus U_2\oplus\cdots\oplus U_{2^n}\in\mathrm{SU(2)^{\oplus2^n}}$ and an error $\varepsilon$, we first apply $\varepsilon$-approximate single-qubit gate synthesis to each of $U_1,U_2,\cdots,U_{2^n}$ and choose a minimum even-T-count Clifford+T circuit with
\begin{equation*}
    m_i=\mathcal{T}_\varepsilon(U_i)
\end{equation*}
for each target block (we can always choose an even-T-count circuit; see Appendix \ref{sec:even_t}). We write these circuits in Matsumoto and Amano's normal form (Equation (\ref{eq:manormal})):
\begin{equation} \label{eq:normal_form}
    \begin{gathered}
        D_\diamond \left(U_i, ~F_i\left(\prod_{j=1}^{m_i} M_{j,i}\right) C_i\right)
        \leq \varepsilon \\[2pt]
        \bigl(i=1,2,\cdots,2^n,\quad F_i\in\{I_2,~H,~SH\}, \quad
        M_{j,i}\in\{TH,~TSH\},\quad C_i\in\mathrm{Clifford}_1\bigr).
    \end{gathered}
\end{equation}
By multiplying a complex number of unit modulus, we can set $C_i$ to be in the form of the product of $S$ and $H$ gates.
We will use the following form, which is equivalent to the one derived by Matsumoto and Amano's normal form:
\begin{equation} \label{eq:my_form}
    D_\diamond\left(
    F_i \left(\prod_{j=1}^{m_i} M_{j,i}\right) C_i
    ,~~
        \zeta^{s_i} F_i \left(\prod_{j=1}^{\frac{m_i+m}{2}} M_{j,i}^\prime\right) H\, \left(\prod_{j=\frac{m_i+m}{2}+1}^{m}M_{j,i}^\prime\right) C_i^\prime 
    \right)=0
    \qquad\qquad
\end{equation}
\begin{equation*}
    \begin{aligned}
    M_{j,i}^\prime &=
    \begin{cases}
        \zeta TH & \text{if } j\leq m_i \text{ and } M_{j,i} = TH, \\
        TSH & \text{if } j\leq m_i \text{ and } M_{j,i} = TSH, \\
        \zeta TH & \text{if } m_i+1\leq j\leq \frac{m_i+m}{2}, \\
        \zeta TH & \text{if }\frac{m_i+m}{2}+1\leq j \text{ and } j\equiv \frac{m_i+m}{2}+2 \mod 3, \\
        TSH & \text{otherwise},
    \end{cases}, \\[1ex]
    C_i^\prime &=
    \begin{cases}
        HZC_i & \text{if } \frac{m-m_i}{2}\equiv 1 \mod 3, \\
        HXC_i & \text{if } \frac{m-m_i}{2}\equiv 2 \mod 3, \\
        HC_i & \text{if } \frac{m-m_i}{2}\equiv 0 \mod 3
    \end{cases},
    \end{aligned}
\end{equation*}
Here $m=\max_i m_i$. We initially assume $m\geq2$; the case
$m=0$ is treated below. Choose $s_i\in\{0,\ldots,7\}$ so that
the second argument, denoted by $U_i'$, satisfies
\begin{equation}\label{eq:aligned_block_sign}
    \det U_i'=1,
    \qquad \operatorname{Re}\operatorname{tr}(U_i^\dagger U_i')\geq0.
\end{equation}
Such a choice exists because each $m_i$ is even: before phase
correction the determinant is a fourth root of unity, and the
two determinant-one lifts differ by a sign. This is the
Frobenius-nearest determinant-one lift, and it fixes the
relative phases of the blocks.
Equation (\ref{eq:my_form}) holds because
\begin{eqnarray}
    (TH)^\dagger &=& HTSH\cdot HZ, \\
    (THTH)^\dagger &=& \zeta^{-1} HTSHTH\cdot HX, \\
    (THTHTH)^\dagger &=& \zeta^{-1} HTSHTHTSH\cdot H.
\end{eqnarray}
Each aligned word contains one inserted $H$, at position
$h_i=(m+m_i)/2$. Let
\begin{equation}\label{eq:occupied_h_layers}
    J=\{h_i:1\leq i\leq 2^n\},\qquad
    K=|J|\leq\min\{2^n,m/2+1\}.
\end{equation}
Layers in which every block is the identity are omitted.
The instruction register consequently has
$b=m+K+O(1)=O(m+1)$ bits, as described below.
We use registers $\mathcal A,\mathcal B,\mathcal C$ consisting of
$n$ control qubits, one target qubit, and $b$ clean instruction
qubits. The workspace $\mathcal D$ contains $q=O(n)$ additional
clean scratch qubits and $O(\sqrt{2^n(m+1)})$ borrowed dirty qubits.
A dirty qubit may initially be in an arbitrary state, possibly
entangled with other registers, and must be restored exactly.
The circuit has three steps:
\begin{enumerate}
    \item compute the instruction string into $\mathcal C$ by a Boolean oracle;
    \item apply the instruction-controlled gates to $\mathcal B$;
    \item uncompute the instruction string, restoring every ancillary qubit.
\end{enumerate}
Figure~\ref{fig:circuit-construction-overview} summarizes this
construction. For a control label $i\in\{0,\ldots,2^n-1\}$, the
instruction string $g(i)$ specifies the operation $U_{i+1}'$.

\begin{figure}[t]
    \centering
    \begin{tikzpicture}[
        x=1cm,
        y=1cm,
        font=\small,
        wire/.style={-{Latex[length=2mm]}, thick},
        operation/.style={draw, rounded corners, fill=blue!7,
            minimum width=1.8cm, minimum height=3cm},
        controloperation/.style={draw, rounded corners, fill=orange!10,
            text width=3.6cm, minimum height=2cm,
            inner xsep=4pt, inner ysep=6pt, align=center},
        workspace/.style={dashed, thick, gray}
    ]
        \draw[wire] (0,3.6) -- (12.8,3.6);
        \draw[wire] (0,1.2) -- (12.8,1.2);
        \draw[wire] (0,0) -- (12.8,0);

        \node[left] at (0,3.6) {$\mathcal{A}:|i\rangle$};
        \node[left] at (0,2.4) {$\mathcal{B}:|\psi\rangle$};
        \node[left] at (0,1.2) {$\mathcal{C}:|0^b\rangle$};
        \node[left] at (0,0) {$\mathcal{D}:|w\rangle$};

        \node[operation] (compute) at (4.1,2.4) {};
        \node[operation] (uncompute) at (11.3,2.4) {};
        \node[align=center] at (4.1,3.25) {Boolean\\layer $U_g$};
        \node[align=center] at (11.3,3.25) {inverse\\layer $U_g^\dagger$};

        \foreach \x in {4.1,11.3} {
            \draw[gray!25, line width=4.8pt] (\x-1.02,2.4) -- (\x+1.02,2.4);
            \draw[white, line width=3.6pt] (\x-1.02,2.4) -- (\x+1.02,2.4);
        }
        \draw[wire] (0,2.4) -- (12.8,2.4);

        \node[controloperation] (select) at (8,1.8)
            {instruction-controlled\\gates\\[3pt]
             {\footnotesize $g(i)_j=0$: $P_j$\quad $g(i)_j=1$: $Q_j$}\\[3pt]
             $\prod_j(P_j\oplus Q_j)$};

        \draw[workspace] (4.1,0) -- (compute.south);
        \draw[workspace] (11.3,0) -- (uncompute.south);
        \node[below=3pt, font=\scriptsize, text=gray] at (4.1,0) {workspace};
        \node[below=3pt, font=\scriptsize, text=gray] at (11.3,0) {workspace};

        \node[above=3pt, font=\footnotesize, inner sep=0pt] at (5.55,1.2)
            {$|g(i)\rangle$};

        \draw[thick] (1.42,3.48) -- (1.58,3.72);
        \draw[thick] (1.42,1.08) -- (1.58,1.32);
        \draw[thick] (1.42,-0.12) -- (1.58,0.12);
        \node[above=4pt, font=\footnotesize] at (1.5,3.6) {$n$};
        \node[above=4pt, font=\footnotesize] at (1.5,1.2)
            {$b=m+K+O(1)$};
        \node[above=4pt, font=\footnotesize] at (1.5,0)
            {$q+O(\sqrt{2^n(m+1)})$};

        \node[right] at (12.8,3.6) {$|i\rangle$};
        \node[right] at (12.8,2.4) {$U_{i+1}^\prime|\psi\rangle$};
        \node[right] at (12.8,1.2) {$|0^b\rangle$};
        \node[right] at (12.8,0) {$|w\rangle$};
    \end{tikzpicture}
    \caption{Overview of the circuit construction for a fixed control value $i$. The Boolean layer computes the classical instruction string $g(i)$, which selects the single-qubit operation applied to $\mathcal{B}$ at each step. The target-register wire passes over the Boolean-layer blocks because these layers do not act on $\mathcal{B}$. The inverse layer uncomputes the instruction string, returning the instruction register and workspace to their initial states. Here $|w\rangle=|0^q\rangle|d\rangle$ includes $q=O(n)$ clean scratch qubits as well as the borrowed dirty register. Slashes on the wires indicate the numbers of qubits in the respective registers.}
    \label{fig:circuit-construction-overview}
\end{figure}

For a Boolean function $f:\{0,1\}^n\to\{0,1\}^b$,
define its reversible lookup oracle by
\begin{equation}
    U_f|x\rangle|y\rangle=|x\rangle|y\oplus f(x)\rangle,
    \quad 0\leq x<2^n,\quad 0\leq y<2^b.
\end{equation}
\begin{lmm}[Boolean oracle]\label{lmm:boolean}
    For $b\geq1$, $U_f$ can be implemented exactly using
    $O(\sqrt{b2^n})$ T gates, $O(n)$ clean scratch qubits,
    and $O(\sqrt{b2^n})$ dirty qubits, in addition to its input
    and output registers. All scratch and dirty qubits are
    restored exactly.
\end{lmm}
\begin{proof}
    For $b\leq2^n$, use the dirty-workspace SelectSwap construction
    of Low, Kliuchnikov, and Schaeffer
    \cite[Section~2 and Figure~1(d)]{Low2018TradingTG}, choosing
    a power of two $\lambda$ within a constant factor of
    $\sqrt{2^n/b}$. Its T-count is
    $O(2^n/\lambda+b\lambda)$, it borrows $b\lambda$ dirty
    qubits, and its selection circuit uses $O(n)$ clean scratch.
    The compute--copy--uncompute construction implements XOR
    into an arbitrary output string $y$ and restores the borrowed
    register. If $b>2^n$, use the Select circuit without swaps:
    it visits the $2^n$ addresses with $O(2^n)$ T gates and
    $O(n)$ clean scratch, copying each stored word by CNOT gates.
    Since $2^n\leq\sqrt{b2^n}$, the same bound holds.
\end{proof}

For $m\geq2$, define $H_{j,i}=H$ if $j=h_i$ and
$H_{j,i}=I_2$ otherwise. The aligned product is
\begin{equation}\label{eq:decompose}
    \begin{aligned}
    \bigoplus_{i=1}^{2^n}U_i'
      &=\left(\bigoplus_{i=1}^{2^n}\zeta^{s_i}I_2\right)
        \left(\bigoplus_{i=1}^{2^n}F_i\right)\\
      &\quad\times\prod_{j=1}^{m}
        \left[
          \left(\bigoplus_{i=1}^{2^n}M_{j,i}'\right)
          \left(\bigoplus_{i=1}^{2^n}H_{j,i}\right)
        \right]
        \left(\bigoplus_{i=1}^{2^n}C_i'\right).
    \end{aligned}
\end{equation}
Only the $K$ positions $j\in J$ require an $H$ instruction.
All factors are therefore encoded by $b=m+K+b_0$ binary
choices, where $b_0$ is an absolute constant accounting for
$F_i$, $C_i'$, and $\zeta^{s_i}$. More explicitly, fix a list
of pairs $(P_r,Q_r)$, independent of the block label, with
\begin{equation}\label{eq:pq}
    (P_r,Q_r)\in
    \{(TSH,\zeta TH),(I_2,H),(I_2,S),
       (I_2,\zeta I_2),(I_2,iI_2),(I_2,-I_2)\}.
\end{equation}
There are $m$ pairs of the first type, $K$ inserted-$H$ pairs,
and $b_0$ endpoint pairs. The endpoint Cliffords belong to a
finite group, so a fixed finite binary template over $H$ and
$S$ suffices for every block. Let $\gamma_{r,i}\in\{0,1\}$
choose $P_r$ or $Q_r$ in this template, ordered so that
\begin{equation}
    U_i'=R_{1,i}\cdots R_{b,i},\qquad
    R_{r,i}=
    \begin{cases}P_r,&\gamma_{r,i}=0,\\Q_r,&\gamma_{r,i}=1.
    \end{cases}
\end{equation}
Define the zero-indexed lookup function by
$g(x)_r=\gamma_{r,x+1}$ for $0\leq x<2^n$ and $1\leq r\leq b$.
After computing $g(x)$ into $\mathcal C$, apply
$P_r\oplus Q_r$ using its $r$th bit as control, in the order
$r=b,b-1,\ldots,1$. The resulting operation on $\mathcal B$
is exactly $U_{x+1}'$. Applying $U_g^\dagger$ then restores
$\mathcal C$ to $|0^b\rangle$ and restores the workspace.

If $m=0$, every chosen single-qubit approximant is a Clifford.
Choose its determinant-one sign-aligned lift as in
Equation~\eqref{eq:aligned_block_sign}, and use only the finite
endpoint template, with $b=b_0$ and no $M'$ or inserted-$H$
layers. This avoids any undefined index $M_{0,i}'$ and costs
$O(\sqrt{2^n})$ T gates including the lookup and its inverse.

Finally, sign alignment is needed to pass from single-block
channel errors to the controlled operation. The eigenvalues of
$U_i^\dagger U_i'$ are $e^{\pm i\alpha_i}$, where
$|\alpha_i|\leq\pi/2$ and $|\sin\alpha_i|\leq\varepsilon$.
Writing $\alpha=\max_i|\alpha_i|$, all eigenvalues of the
relative controlled unitary lie on the arc
$\{e^{it}:|t|\leq\alpha\}$, and both endpoints occur.
Their convex hull is at distance $\cos\alpha$ from the origin.
The unitary-channel formula therefore gives
\begin{equation}\label{eq:aligned_direct_sum_error}
    D_\diamond\!\left(\bigoplus_iU_i,\bigoplus_iU_i'\right)
      =\sin\alpha
      =\max_iD_\diamond(U_i,U_i')\leq\varepsilon.
\end{equation}
The equality depends on the common determinant-one sign
choice; it does not hold for independently phased blocks.

\subsection{Lower bound with clean ancillae}
\label{ssec:lower_bound}

The construction in Section~\ref{ssec:construction} preserves the
controlled structure of the target unitary. To establish its
asymptotic optimality, however, we must also consider approximating
circuits that do not preserve this structure. 

For $U\in\mathrm{U}(2^{n+1})$, let
$\mathcal{T}^{\mathrm{clean}}_\varepsilon(U)$ denote the minimum
T-count of a Clifford+T circuit $W$ acting on $n+1+a$ qubits,
where the minimum is taken over all integers $a\geq0$ and all
such circuits for which there exists
$V\in\mathrm{U}(2^{n+1})$ satisfying
\begin{equation}
    W\bigl(|\psi\rangle\otimes|0^a\rangle\bigr)
    =
    V|\psi\rangle\otimes|0^a\rangle
    \qquad
    \text{for every }|\psi\rangle,
\end{equation}
and
\begin{equation}
    D_\diamond(U,V)\leq\varepsilon.
\end{equation}
Thus, all ancillae must be returned exactly to their initial state,
whereas the induced unitary $V$ need not be block diagonal.
No parity restriction is imposed on the T-count, and measurements
and postselection are not allowed.

\begin{lmm}[T-count lower bound with clean ancillae]
\label{lmm:clean_ancilla_lower_bound}
    Let $n\geq1$, and let
    $U_1,\ldots,U_{2^n}$ be independent Haar-random elements of
    $\mathrm{SU}(2)$. Set
    \begin{equation}
        U=\bigoplus_{i=1}^{2^n}U_i.
    \end{equation}
    For an integer $t\geq2$, define
    \begin{equation}
        s=\left\lceil\frac{t}{2}\right\rceil+n+1,
        \qquad
        \gamma=(\sqrt{2})^s.
    \end{equation}
    If $0<\varepsilon<1$ and
    \begin{equation}
        64\gamma^4\varepsilon^3<1,
    \end{equation}
    then
    \begin{equation}
    \label{eq:clean_ancilla_tail}
        \Pr\!\left[
            \mathcal{T}^{\mathrm{clean}}_\varepsilon(U)\leq t
        \right]
        \leq
        2^{o(s)}
        \left(
            \gamma^4\varepsilon^3
            +2^{2n}\gamma^{12}\varepsilon^9
            +2^{2n}\gamma^{20}\varepsilon^{15}
        \right).
    \end{equation}
    The bound is independent of the number of clean ancillae.
\end{lmm}

In particular, fix $n$ and $\delta\in(0,3)$, and take
$t=\lfloor(3-\delta)\log_2(1/\varepsilon)\rfloor$.
The condition of Lemma~\ref{lmm:clean_ancilla_lower_bound}
then holds for all sufficiently small $\varepsilon$, and
Equation~\eqref{eq:clean_ancilla_tail} gives
\begin{equation}
\label{eq:clean_ancilla_asymptotic}
    \Pr\!\left[
        \mathcal{T}^{\mathrm{clean}}_\varepsilon(U)
        \leq
        (3-\delta)\log_2(1/\varepsilon)
    \right]
    \leq
    \varepsilon^{\delta-o(1)}
    \qquad
    (\varepsilon\to0).
\end{equation}
Consequently, for every fixed $\delta>0$, a T-count below
$(3-\delta)\log_2(1/\varepsilon)$ is insufficient with high
probability. Allowing non-block-diagonal approximating unitaries
or arbitrarily many clean ancillae therefore cannot reduce the
leading coefficient below $3$.

\subsubsection{Auxiliary lemmas}
\label{sssec:clean_ancilla_auxiliary}

We collect the arithmetic and geometric estimates used in the
proof of Lemma~\ref{lmm:clean_ancilla_lower_bound}.
The first two lemmas describe the arithmetic and block constraints
on an approximating unitary. The next three count integral vectors
and matrices satisfying prescribed norm conditions, and the last
bounds the Haar measure of compatible rotations.
Their proofs appear in the same order in
Appendix~\ref{app:clean_auxiliary_proofs}.

Throughout the remainder of Section~\ref{ssec:lower_bound},
use $s$ and $\gamma$ as defined in
Lemma~\ref{lmm:clean_ancilla_lower_bound}.
We work over
\begin{equation}
    k=\mathbb{Q}(\sqrt{2}),
    \qquad
    E=k(i)=\mathbb{Q}(\zeta),
    \qquad
    \zeta=e^{i\pi/4},
\end{equation}
with rings of integers
\begin{equation}
    \mathcal{O}_k=\mathbb{Z}[\sqrt{2}],
    \qquad
    \mathcal{O}_E=\mathbb{Z}[\zeta].
\end{equation}
We also write
\begin{equation}
    \mathcal{R}=\mathcal{O}_E[1/\sqrt{2}].
\end{equation}
Complex conjugation is denoted by an overline.
Let $\sigma$ be the automorphism of $E$ defined by
\begin{equation}
    \sigma(\sqrt{2})=-\sqrt{2},
    \qquad
    \sigma(i)=i.
\end{equation}
It is applied entrywise to vectors and matrices.
For $x\in k$ and a nonzero integral ideal
$\mathfrak{a}\subseteq\mathcal{O}_k$, write
\begin{equation}
    N_k(x)=|x\sigma(x)|,
    \qquad
    N_k(\mathfrak{a})=
    |\mathcal{O}_k/\mathfrak{a}|.
\end{equation}

We use $\|\cdot\|_{\mathrm{op}}$ for the operator norm.
For $M\in M_{2^{n+1}}(\mathbb{C})$, we denote its
$(i,j)$ block by
\begin{equation*}
    \blk{M}{i}{j}:=
    \begin{pmatrix}
        M_{2i-1,2j-1} & M_{2i-1,2j}\\
        M_{2i,2j-1} & M_{2i,2j}
    \end{pmatrix}
    \in M_2(\mathbb{C}),
    \qquad 1\leq i,j\leq 2^n.
\end{equation*}
Thus, $\blk{M}{i}{j}$ is the $2\times2$ submatrix formed
by rows $2i-1,2i$ and columns $2j-1,2j$.
For the controlled target $U=\bigoplus_{i=1}^{2^n}U_i$,
we have $\blk{U}{i}{i}=U_i$.

For the arithmetic arguments, circuit matrices are represented
as products of the specified $H$, $S$, CNOT, and $T$ matrices,
rather than multiplied by arbitrary global phases.
This convention does not affect the diamond distance or T-count.


The first lemma turns a bound on the T-count into a common
denominator bound for the logical unitary. Thus multiplying by
$\gamma$ gives an integral matrix of fixed size $2^{n+1}$, regardless
of how many clean ancillae the circuit uses. Its Galois and
determinant constraints further restrict the matrices to be counted.

\begin{lmm}[Arithmetic constraints with clean ancillae]
\label{lmm:clean_arithmetic}
    Let $V\in\mathrm{U}(2^{n+1})$ be implemented exactly by a
    Clifford+T circuit with T-count at most $t$ and any number
    of clean ancillae. Then
    \begin{equation}
    \label{eq:clean_ptm_denominator}
        (\sqrt{2})^t\widehat{V}
        \in M_{2^{2n+2}}(\mathcal{O}_k),
    \end{equation}
    and
    \begin{equation}
    \label{eq:clean_matrix_denominator}
        V\in M_{2^{n+1}}(\mathcal{R}),
        \qquad
        \gamma V\in M_{2^{n+1}}(\mathcal{O}_E).
    \end{equation}
    Moreover,
    \begin{equation}
    \label{eq:clean_galois_determinant}
        \sigma(V)\in\mathrm{U}(2^{n+1}),
        \qquad
        \det(V)^8=1.
    \end{equation}
    These conclusions are independent of the number of ancillae.

    More generally, every unitary matrix $Y$ with entries in
    $\mathcal{R}$ satisfies $\det(Y)^8=1$.
\end{lmm}

The proof is given in Appendix~\ref{app:proof_clean_arithmetic}.

We next translate diamond-distance approximation into constraints
on the individual $2\times2$ blocks. The determinant restriction
above limits the global phase to finitely many choices, while the
following bounds control off-diagonal leakage, the singular values
of diagonal blocks, and cancellation between opposite blocks.

\begin{lmm}[Block constraints from diamond-distance approximation]
\label{lmm:clean_block_geometry}
    Let
    $U=\bigoplus_{i=1}^{2^n}U_i$ with
    $U_i\in\mathrm{SU}(2)$, and let
    $V\in\mathrm{U}(2^{n+1})$ satisfy $\det(V)^8=1$.
    Suppose that $0<\varepsilon<1$ and
    $D_\diamond(U,V)\leq\varepsilon$.
    Then there exists $\eta\in\mathbb{C}$ satisfying
    $\eta^{2^{n+4}}=1$ such that
    \begin{equation}
    \label{eq:clean_phase_alignment}
        \|V-\eta U\|_{\mathrm{op}}\leq2\varepsilon.
    \end{equation}
    Furthermore,
    \begin{equation}
    \label{eq:clean_block_leakage}
        \|\blk{V}{i}{j}\|_{\mathrm{op}}\leq\varepsilon
        \qquad (i\neq j),
    \end{equation}
    and
    \begin{equation}
    \label{eq:clean_diagonal_singular_values}
        (1-\varepsilon^2)I_2
        \preceq \blk{V}{i}{i}^\dagger \blk{V}{i}{i}
        \preceq I_2.
    \end{equation}
    For the same choice of $\eta$, we also have
    \begin{equation}
    \label{eq:clean_block_cancellation}
        \|U_i^\dagger \blk{V}{i}{j}
          +\eta^2\blk{V}{j}{i}^\dagger U_j\|_{\mathrm{op}}
        \leq4\varepsilon^2
        \qquad (i\neq j).
    \end{equation}
\end{lmm}

The proof is given in Appendix~\ref{app:proof_clean_block_geometry}.

To count the integral blocks allowed by these constraints, we
first need estimates for scalar and two-component norm equations.
Here an element of $\mathcal{O}_k$ is \emph{totally positive} if
it and its image under $\sigma$ are both positive.
The next lemma bounds the number of possible coordinates and
column vectors with a prescribed squared norm.

\begin{lmm}[Counting solutions of norm equations]
\label{lmm:clean_norm_equations}
    For every $\nu>0$, there exists $C_\nu>0$ such that,
    for all totally positive
    $\alpha,\beta\in\mathcal{O}_k$,
    \begin{align}
        \#\{z\in\mathcal{O}_E:
             z\overline{z}=\alpha\}
        &\leq C_\nu N_k(\alpha)^\nu,
        \label{eq:clean_norm_one}\\
        \#\{(u_1,u_2)\in\mathcal{O}_E^2:
             u_1\overline{u_1}
             +u_2\overline{u_2}=\beta\}
        &\leq C_\nu N_k(\beta)^{1+\nu}.
        \label{eq:clean_norm_two}
    \end{align}
\end{lmm}

The proof is given in Appendix~\ref{app:proof_clean_norm_equations}.

The Gram matrix $M^\dagger M$ records the squared norms and
inner products of the columns of $M$. The next estimate counts
matrices with the same Gram matrix; its dependence on common
arithmetic factors is expressed through the following content ideal.
For $G\in M_2(\mathcal{O}_E)$, define its
$\mathcal{O}_k$-content as follows.
Using
$\mathcal{O}_E=\mathcal{O}_k\oplus\zeta\mathcal{O}_k$,
write
\begin{equation}
    G_{ij}=a_{ij}+b_{ij}\zeta,
    \qquad
    a_{ij},b_{ij}\in\mathcal{O}_k,
\end{equation}
and set
\begin{equation}
    \mathfrak{c}(G)
    =
    (a_{ij},b_{ij}:1\leq i,j\leq2).
\end{equation}
Equivalently, $\mathfrak{c}(G)$ is the smallest ideal
$\mathfrak{a}\subseteq\mathcal{O}_k$ such that
$G_{ij}\in\mathfrak{a}\mathcal{O}_E$ for every $i,j$.

\begin{lmm}[Counting matrices with a fixed Gram matrix]
\label{lmm:clean_fixed_gram}
    Fix $A>0$.
    Let $G\in M_2(\mathcal{O}_E)$ be Hermitian, with
    both $G$ and $\sigma(G)$ positive definite.
    If
    \begin{equation}
        N_k(\det G)\leq2^{As},
    \end{equation}
    then
    \begin{equation}
    \label{eq:clean_fixed_gram}
        \#\{M\in M_2(\mathcal{O}_E):M^\dagger M=G\}
        \leq
        2^{o(s)}N_k(\mathfrak{c}(G)).
    \end{equation}
    The factor $2^{o(s)}$ is uniform over $G$, with
    dependence on $A$ permitted.
\end{lmm}

The proof is given in Appendix~\ref{app:fixed_gram}.

For counting blocks, it is also useful to retain only the trace
of the Gram matrix and the determinant of the block.
The following estimate sums over compatible Gram matrices and
also covers singular matrices with $\det M=0$.

\begin{lmm}[Counting matrices with fixed trace and determinant]
\label{lmm:clean_trace_determinant}
    Fix $A>0$.
    Let $\alpha\in\mathcal{O}_k$ be totally positive and
    satisfy $N_k(\alpha)\leq2^{As}$.
    Then, uniformly over $z\in\mathcal{O}_E$,
    \begin{equation}
    \label{eq:clean_trace_determinant}
        \#\left\{
            M\in M_2(\mathcal{O}_E):
            \operatorname{tr}(M^\dagger M)=\alpha,\,
            \det M=z
        \right\}
        \leq2^{o(s)}N_k(\alpha).
    \end{equation}
    The factor $2^{o(s)}$ is uniform over $\alpha$ and $z$, with
    dependence on $A$ permitted.
\end{lmm}

The proof is given in Appendix~\ref{app:proof_clean_trace_determinant}.

Finally, we need to combine the arithmetic counts with a measure
estimate. The next lemma bounds the probability that a fixed
matrix, after independent left and right $\mathrm{SU}(2)$
rotations, lies near a prescribed matrix.

\begin{lmm}[Small-ball bound for two-sided SU(2) actions]
\label{lmm:clean_small_ball}
    Let $P\in M_2(\mathbb{C})\setminus\{0\}$ and
    $F\in M_2(\mathbb{C})$.
    Write the singular values of $P$ as $r\geq u\geq0$
    and set $b=r-u$.
    For independent Haar-random
    $L,R\in\mathrm{SU}(2)$ and every $e>0$,
    \begin{equation}
    \label{eq:clean_small_ball}
        \Pr\!\left[
            \|LPR-F\|_{\mathrm{op}}\leq e
        \right]
        \leq
        C\min\{1,(e/r)^3\}
         \min\{1,(e/b)^2\},
    \end{equation}
    where $C$ is an absolute constant.
    When $b=0$, the second factor is defined to be $1$.
\end{lmm}

The proof is given in Appendix~\ref{app:proof_clean_small_ball}.

These lemmas separate the lower-bound argument into arithmetic
restrictions, counts of admissible blocks, and geometric estimates
for the measure of nearby targets. In particular, the arithmetic
constraints are imposed on the logical unitary, so the resulting
counting problem does not grow with the number of clean ancillae.

\subsubsection{Proof of the lower bound}
\label{sssec:clean_ancilla_proof}

\begin{proof}[Proof of Lemma~\ref{lmm:clean_ancilla_lower_bound}]
Let $\mathcal{H}=\mathrm{SU}(2)^{\oplus 2^n}$, and let $\mu_n$
be its normalized Haar measure.
Let $\mathcal{V}_t$ be the set of logical unitaries implemented
exactly by Clifford+T circuits with T-count at most $t$ and
any number of clean ancillae, using the matrix representatives
specified above.
For $S\subseteq\mathrm{U}(2^{n+1})$, define
\begin{equation}
    \mathcal{N}_\varepsilon(S)
    =\{U\in\mathcal{H}:\text{there exists }V\in S
        \text{ with }D_\diamond(U,V)\leq\varepsilon\}.
\end{equation}
Then
\begin{equation}
    \Pr[\mathcal{T}^{\mathrm{clean}}_\varepsilon(U)\leq t]
    =\mu_n(\mathcal{N}_\varepsilon(\mathcal{V}_t)).
\end{equation}
We split $\mathcal{V}_t$ into
\begin{equation}
    \mathcal{V}_t^{(1)}
    =\{V\in\mathcal{V}_t:\blk{V}{1}{1}\in\mathrm{U}(2)\},
    \qquad
    \mathcal{V}_t^{(2)}
    =\mathcal{V}_t\setminus\mathcal{V}_t^{(1)}.
\end{equation}
All subexponential factors below are uniform over the
arithmetic labels being counted.

\medskip
\noindent\emph{Case 1: the first diagonal block is unitary.}
Let $\mathcal{Y}=\{\blk{V}{1}{1}:V\in\mathcal{V}_t^{(1)}\}$.
By Lemma~\ref{lmm:clean_arithmetic}, every $Y\in\mathcal{Y}$
satisfies $\det(Y)^8=1$ and has first column
$\gamma^{-1}(x_1,x_2)^{\mathsf T}$ with
\begin{equation}
    x_1,x_2\in\mathcal{O}_E,
    \qquad x_1\overline{x_1}+x_2\overline{x_2}=\gamma^2.
\end{equation}
Since $N_k(\gamma^2)=\gamma^4$,
Lemma~\ref{lmm:clean_norm_equations}
bounds the number of possible first columns by
$2^{o(s)}\gamma^4$.
A $2\times2$ unitary is determined by its first column and
determinant, and there are at most eight possible determinants.
Thus $\#\mathcal{Y}\leq2^{o(s)}\gamma^4$.

If $\blk{V}{1}{1}=Y$ is unitary, the remaining blocks in the
first block row and column vanish.
Restricting the input to the first control block therefore gives
$D_\diamond(U_1,Y)\leq D_\diamond(U,V)$.
For each fixed $Y\in\mathrm{U}(2)$, the set of
$U_1\in\mathrm{SU}(2)$ satisfying
$D_\diamond(U_1,Y)\leq\varepsilon$
has Haar measure $O(\varepsilon^3)$.
The union bound yields
\begin{equation}
\label{eq:clean_unitary_block_contribution}
    \mu_n(\mathcal{N}_\varepsilon(\mathcal{V}_t^{(1)}))
    \leq2^{o(s)}\gamma^4\varepsilon^3.
\end{equation}

\medskip
\noindent\emph{Case 2: the first diagonal block is not unitary.}
Unitarity implies that every $V\in\mathcal{V}_t^{(2)}$ has
$\blk{V}{1}{j}\neq0$ for some $j\in\{2,\ldots,2^n\}$.
Fix such an index $j$ and write
\begin{equation}
    \mathcal{V}_{t,j}
    =\{V\in\mathcal{V}_t:\blk{V}{1}{j}\neq0\},
    \qquad
    P=\blk{V}{1}{j},\quad Q=\blk{V}{j}{1}.
\end{equation}
We first derive necessary conditions on $P,Q$ for
$D_\diamond(U,V)\leq\varepsilon$.

\smallskip
\noindent\emph{Arithmetic restrictions on the block pair.}
Lemmas~\ref{lmm:clean_arithmetic}
and~\ref{lmm:clean_block_geometry} give
\begin{equation}
    \gamma P,\gamma Q\in M_2(\mathcal{O}_E),\qquad
    \|P\|_{\mathrm{op}},\|Q\|_{\mathrm{op}}\leq\varepsilon,
    \qquad
    \|\sigma(P)\|_{\mathrm{op}},
    \|\sigma(Q)\|_{\mathrm{op}}\leq1,
\end{equation}
and a phase $\eta$ satisfying $\eta^{2^{n+4}}=1$ such that
\begin{equation}
\label{eq:clean_pair_relation}
    \|U_1^\dagger P+\eta^2Q^\dagger U_j\|_{\mathrm{op}}
    \leq4\varepsilon^2.
\end{equation}
Put $A=U_1^\dagger P$ and $B=-\eta^2Q^\dagger U_j$.
The signed difference
\begin{equation}
    h=\operatorname{tr}(P^\dagger P)
       -\operatorname{tr}(Q^\dagger Q)
    \in\gamma^{-2}\mathcal{O}_k
\end{equation}
satisfies
\begin{equation}
    |h|\leq
    2(\|A\|_{\mathrm{op}}+\|B\|_{\mathrm{op}})
      \|A-B\|_{\mathrm{op}}
    \leq16\varepsilon^3,
    \qquad |\sigma(h)|\leq4.
\end{equation}
If $h\neq0$, then
\begin{equation}
    1\leq N_k(\gamma^2h)
    =\gamma^4|h\sigma(h)|
    \leq64\gamma^4\varepsilon^3<1,
\end{equation}
a contradiction.
Hence the two traces match exactly, and we may define
\begin{equation}
\label{eq:clean_common_trace_label}
    \alpha=\gamma^2\operatorname{tr}(P^\dagger P)
           =\gamma^2\operatorname{tr}(Q^\dagger Q)
    \in\mathcal{O}_k.
\end{equation}
Let $r\geq u\geq0$ be the singular values of $P$, and set $b=r-u$.
Since $P\neq0$, both $\alpha$ and $\sigma(\alpha)$ are positive,
and
\begin{equation}
    0<\alpha\leq2\gamma^2r^2,\qquad
    0<\sigma(\alpha)\leq2\gamma^2.
\end{equation}
Thus $1\leq N_k(\alpha)\leq4\gamma^4r^2$, which gives
\begin{equation}
\label{eq:clean_nonzero_block_scale}
    r_*:=\frac{1}{2\gamma^2}\leq r\leq\varepsilon.
\end{equation}
If $\varepsilon<r_*$, this case contributes nothing.
We may therefore assume $\varepsilon\geq r_*$; in particular,
\begin{equation}
\label{eq:clean_logarithmic_range}
    \log_2(1/\varepsilon)\leq s+1.
\end{equation}

Writing $A=(a_1,a_2)$ and $B=(b_1,b_2)$ in terms of their
columns, multilinearity of the determinant gives
\begin{equation}
    \det A-\det B
    =
    \det(a_1-b_1,a_2)
    +
    \det(b_1,a_2-b_2).
\end{equation}
Since $|\det(x,y)|\leq\|x\|_2\|y\|_2$ for
$x,y\in\mathbb{C}^2$, we obtain
\begin{equation}\label{eq:detAB}
    |\det A-\det B|
    \leq
    (\|A\|_{\mathrm{op}}+\|B\|_{\mathrm{op}})
    \|A-B\|_{\mathrm{op}}.
\end{equation}

Define the determinant labels
\begin{equation}
    z=\gamma^2\det P,\qquad z'=\gamma^2\det Q,
    \qquad z,z'\in\mathcal{O}_E.
\end{equation}
Since $\det U_1=\det U_j=1$,
Equation~\eqref{eq:clean_pair_relation} and \eqref{eq:detAB}
therefore implies
\begin{equation}
\label{eq:clean_determinant_matching}
    |z-\eta^4\overline{z'}|
    \leq8\gamma^2\varepsilon^3.
\end{equation}

\smallskip
\noindent\emph{A common event for all completions of a block pair.}
For fixed $P,Q,\eta$, define
\begin{equation}
    \mathcal{E}_j(P,Q,\eta)
    =\{U\in\mathcal{H}:
      \|U_1^\dagger P U_j^\dagger+\eta^2Q^\dagger\|_{\mathrm{op}}
      \leq4\varepsilon^2\}.
\end{equation}
Since $U_1$ and $U_j$ are independent Haar-random elements
of $\mathrm{SU}(2)$, Lemma~\ref{lmm:clean_small_ball} gives
\begin{equation}
\label{eq:clean_pair_event_measure}
    \mu_n(\mathcal{E}_j(P,Q,\eta))
    \leq C(\varepsilon^2/r)^3
                \min\{1,(\varepsilon^2/b)^2\}.
\end{equation}
Here and below the second factor is $1$ when $b=0$.
It therefore suffices to count block pairs, rather than their
possible completions to full circuit unitaries.

\smallskip
\noindent\emph{Counting block pairs in a singular-value range.}
Fix $r_*\leq r_0\leq\varepsilon$, $B_0\geq0$, and
$\eta^{2^{n+4}}=1$.
Let $\mathcal{Z}(r_0,B_0,\eta)$ be the set of pairs
$P,Q\in\gamma^{-1}M_2(\mathcal{O}_E)$ satisfying
\begin{equation}
\begin{gathered}
    P\neq0,\qquad
    \|\sigma(P)\|_{\mathrm{op}},\|\sigma(Q)\|_{\mathrm{op}}\leq1,
    \qquad
    r_0\leq r<2r_0,\quad b\leq B_0,\\
    \alpha:=\gamma^2\operatorname{tr}(P^\dagger P)
            =\gamma^2\operatorname{tr}(Q^\dagger Q),\\
    z:=\gamma^2\det P,\quad z':=\gamma^2\det Q,\qquad
    |z-\eta^4\overline{z'}|\leq8\gamma^2\varepsilon^3.
\end{gathered}
\end{equation}
Here $r\geq u\geq0$ are the singular values of $P$, and $b=r-u$.
The set may include pairs that do not extend to a circuit
unitary; this only increases the upper bound.

The common trace label satisfies
\begin{equation}
    0<\alpha\leq8\gamma^2r_0^2,\qquad
    0<\sigma(\alpha)\leq2\gamma^2.
\end{equation}
Distinct possible values of $\alpha$ are separated by at least
$1/(2\gamma^2)$ at the first embedding, because their difference
is a nonzero algebraic integer with absolute norm at least $1$.
Hence there are $O(\gamma^4r_0^2)$ possible values of $\alpha$;
the additive constant is absorbed since $r_0\geq r_*$.

Next, define
\begin{equation}
    \Delta=\alpha^2-4z\overline{z}
           =\gamma^4(r^2-u^2)^2\in\mathcal{O}_k.
\end{equation}
The singular-value and conjugate-norm bounds give
\begin{equation}
    0\leq\Delta\leq16\gamma^4r_0^2B_0^2,\qquad
    0\leq\sigma(\Delta)\leq4\gamma^4.
\end{equation}
The same spacing argument bounds the number of values of
$\Delta$ by $O(1+\gamma^8r_0^2B_0^2)$.

For fixed $\alpha,\Delta$, the label $z$ satisfies
\begin{equation}
    z\overline{z}=\frac{\alpha^2-\Delta}{4}.
\end{equation}
For a nonzero right-hand side that admits a solution, its norm
is at most $N_k(\alpha)^2/16=2^{O(s)}$.
Lemma~\ref{lmm:clean_norm_equations} therefore gives at most
$2^{o(s)}$ values of $z$.
If the right-hand side is zero, the only value is $z=0$.

For fixed $z,\eta$, there is at most one possible $z'$.
Indeed, $|\sigma(z')|\leq\gamma^2$.
For distinct $z'_1,z'_2\in\mathcal{O}_E$ with this bound,
integrality of their nonzero difference gives
\begin{equation}
    1\leq|z'_1-z'_2|^2|\sigma(z'_1-z'_2)|^2
    \leq4\gamma^4|z'_1-z'_2|^2.
\end{equation}
Their separation is therefore at least $1/(2\gamma^2)$.
On the other hand, Equation~\eqref{eq:clean_determinant_matching}
places all admissible $z'$ in a disk whose diameter is
\begin{equation}
    16\gamma^2\varepsilon^3
    <\frac{1}{4\gamma^2}
    <\frac{1}{2\gamma^2}.
\end{equation}

Finally, for fixed $\alpha,z,z'$, apply
Lemma~\ref{lmm:clean_trace_determinant} to $\gamma P$
and $\gamma Q$.
Since $r_0\leq\varepsilon<1$, we have
$N_k(\alpha)\leq16\gamma^4r_0^2=2^{O(s)}$.
Each matrix therefore has at most
$2^{o(s)}N_k(\alpha)\leq2^{o(s)}\gamma^4r_0^2$
possible values.
Multiplying all these bounds yields
\begin{equation}
\label{eq:clean_block_pair_count}
    \#\mathcal{Z}(r_0,B_0,\eta)
    \leq2^{o(s)}\gamma^{12}r_0^6
                    (1+\gamma^8r_0^2B_0^2).
\end{equation}

\smallskip
\noindent\emph{Summing the measures.}
Define the dyadic radii
\begin{equation}
    \mathcal{D}
    =\{2^k r_*:k\in\mathbb{Z},\ 0\leq k\leq
        \lfloor\log_2(\varepsilon/r_*)\rfloor\}.
\end{equation}
The necessary conditions above imply the covering
\begin{equation}
\label{eq:clean_pair_cover}
    \mathcal{N}_\varepsilon(\mathcal{V}_{t,j})
    \subseteq
    \bigcup_{\eta^{2^{n+4}}=1}\ \bigcup_{r_0\in\mathcal{D}}
    \bigcup_{(P,Q)\in\mathcal{Z}(r_0,\varepsilon,\eta)}
    \mathcal{E}_j(P,Q,\eta).
\end{equation}
Fix $r_0\in\mathcal{D}$ and $\eta$.
For pairs with $b\leq\varepsilon^2$, use
Equation~\eqref{eq:clean_block_pair_count} with $B_0=\varepsilon^2$
and Equation~\eqref{eq:clean_pair_event_measure}.
The measure of the union of their events is at most
\begin{equation}
    2^{o(s)}
    \bigl(\gamma^{12}r_0^3\varepsilon^6
          +\gamma^{20}r_0^5\varepsilon^{10}\bigr).
\end{equation}
For a range $B<b\leq2B$ with $B\geq\varepsilon^2$, the same
argument with $B_0=2B$ bounds its contribution by
\begin{align}
    &2^{o(s)}\gamma^{12}r_0^6
       (1+\gamma^8r_0^2B^2)
       \frac{\varepsilon^{10}}{r_0^3B^2}\notag\\
    &\hspace{1cm}\leq
       2^{o(s)}
       \bigl(\gamma^{12}r_0^3\varepsilon^6
             +\gamma^{20}r_0^5\varepsilon^{10}\bigr).
\end{align}
For actual approximants, $b\leq r\leq\varepsilon$.
The ranges $b\leq\varepsilon^2$ and
$2^\ell\varepsilon^2<b\leq2^{\ell+1}\varepsilon^2$,
for integers $\ell\geq0$ with $2^\ell\varepsilon^2<\varepsilon$,
therefore cover all possibilities using
$O(1+\log(1/\varepsilon))=O(s+1)$ ranges,
by Equation~\eqref{eq:clean_logarithmic_range}.
This factor is absorbed into $2^{o(s)}$.

For $p=3,5$, the geometric sums over the dyadic radii satisfy
$\sum_{r_0\in\mathcal{D}}r_0^p=O(\varepsilon^p)$.
For each fixed $j,\eta$, the total contribution is consequently
at most
\begin{equation}
    2^{o(s)}
    \bigl(\gamma^{12}\varepsilon^9
          +\gamma^{20}\varepsilon^{15}\bigr).
\end{equation}
There are $2^{n+4}$ possible phases and $2^n-1$ choices of $j$.
Since $\mathcal{V}_t^{(2)}=\bigcup_{j=2}^{2^n}\mathcal{V}_{t,j}$,
we obtain
\begin{equation}
\label{eq:clean_nonunitary_block_contribution}
    \mu_n(\mathcal{N}_\varepsilon(\mathcal{V}_t^{(2)}))
    \leq
    2^{2n+o(s)}
    \bigl(\gamma^{12}\varepsilon^9
          +\gamma^{20}\varepsilon^{15}\bigr).
\end{equation}
Combining Equations~\eqref{eq:clean_unitary_block_contribution}
and~\eqref{eq:clean_nonunitary_block_contribution}
proves Equation~\eqref{eq:clean_ancilla_tail}.
\end{proof}

\subsection{T-count analysis} \label{ssec:tcount_analysis}
We now analyze the T-count of the construction in
Section~\ref{ssec:construction} and combine the resulting upper
bound with Lemma~\ref{lmm:clean_ancilla_lower_bound}
to prove Theorem~\ref{thm:su2cnt}.
We first estimate the maximum and minimum T-counts required
to approximate the individual target blocks.
\begin{lmm}[T-count tails for Haar-random single-qubit targets]
\label{lmm:bound}
    Let $U_1,\ldots,U_k$ be independent Haar-random elements of
    $\mathrm{SU}(2)$, where $k\geq1$, and let
    $0<\varepsilon\leq1/2$.
    For $0<\gamma$ and $\iota>1$, set
    \begin{equation*}
        s=2\left\lceil
            \frac{(3+\gamma)\log_2(1/\varepsilon)+\iota\log_2k}{2}
        \right\rceil,
        \qquad
        r=(3-\gamma)\log_2(1/\varepsilon)-\iota\log_2k.
    \end{equation*}
    There are constants $C_{\gamma,\iota},C>0$, independent of
    $k$ and $\varepsilon$, such that
    \begin{equation} \label{eq:count_max}
        \Pr\!\left[
            \max_{1\leq i\leq k}\mathcal{T}_\varepsilon(U_i)>s
        \right]
        \leq
        C_{\gamma,\iota}\,
        k^{1-\iota}\varepsilon^\gamma
        (\log_2(1/\varepsilon)+\log_2k+1)^2,
    \end{equation}
    and
    \begin{equation} \label{eq:count_min}
        \Pr\!\left[
            \min_{1\leq i\leq k}\mathcal{T}_\varepsilon(U_i)\leq r
        \right]
        \leq Ck^{1-\iota}\varepsilon^\gamma.
    \end{equation}
\end{lmm}
\begin{proof}
    If $r<0$, the lower-tail event is empty.
    For a Haar-random $U\in\mathrm{SU}(2)$, Appendix~\ref{sec:even_t} proves the uniform estimates
    \begin{align}
        \Pr[\mathcal{T}_\varepsilon(U)>s]
        &\leq C\frac{(s+1)^2}{2^s\varepsilon^3}
        &&(s\geq2\text{ even}),
        \label{eq:count_upper}\\
        \Pr[\mathcal{T}_\varepsilon(U)\leq r]
        &\leq C2^r\varepsilon^3
        &&(r\geq0).
        \label{eq:count_lower}
    \end{align}
    Apply a union bound to the $k$ target blocks.
    The choice of $s$ gives $2^{-s}\varepsilon^{-3}\leq
    k^{-\iota}\varepsilon^\gamma$ and
    $s+1=O_{\gamma,\iota}(\log_2(1/\varepsilon)+\log_2k+1)$, proving
    Equation~\eqref{eq:count_max}. Likewise,
    $2^r\varepsilon^3=k^{-\iota}\varepsilon^\gamma$ proves
    Equation~\eqref{eq:count_min}.
\end{proof}

\begin{proof} [Proof of Theorem \ref{thm:su2cnt}.]
    Let $m=\max_i\mathcal T_\varepsilon(U_i)$, as in the
    statement of the theorem, and let $K$ be the number of
    occupied inserted-$H$ positions in
    Equation~\eqref{eq:occupied_h_layers}.
    The identity
    \begin{equation}\label{eq:tshth}
        TSH\oplus\zeta TH
        =(I_2\otimes S)\mathrm{CNOT}(I_2\otimes T)
          \mathrm{CNOT}(I_2\otimes H)
    \end{equation}
    implements each of the $m$ aligned syllable choices with
    one T gate. Each occupied $H$ instruction costs two T gates;
    the finite endpoint template costs $O(1)$. In particular,
    \begin{equation}\label{eq:iiti}
        I_2\oplus(\zeta I_2)=T\otimes I_2.
    \end{equation}
    Thus the instruction-controlled part has cost
    \begin{equation}\label{eq:tcount_control}
        m+2K+O(1),\qquad
        K\leq\min\{2^n,m/2+1\}\leq\sqrt{2^n(m+1)}.
    \end{equation}
    With $b=m+K+O(1)=O(m+1)$, the lookup and its inverse cost
    $O(\sqrt{2^n(m+1)})$ T gates by Lemma~\ref{lmm:boolean}.
    The $m=0$ construction uses only its endpoint template and
    satisfies the same bound. Hence the full construction uses
    \begin{equation}
        m+O(\sqrt{2^n(m+1)})
    \end{equation}
    T gates, $O(m+n+1)$ clean instruction and scratch qubits,
    and $O(\sqrt{2^n(m+1)})$ dirty qubits. Dirty qubits can instead
    be initialized in $|0\rangle$ and returned there exactly,
    so the circuit is also admissible in the clean-ancilla model.
    This proves Equation~\eqref{eq:thm1_instance}.

    Now suppose that $U_1,\ldots,U_{2^n}$ are independent
    Haar-random blocks.
    Lemma~\ref{lmm:uniform_su2_upper} gives
    \begin{equation}
        m\leq6\log_2(1/\varepsilon)+O(\log(\log_2(1/\varepsilon)+2)).
    \end{equation}
    Since $\log_2(1/\varepsilon)\geq1$, there is an absolute constant $B>0$
    such that $m+1\leq B\log_2(1/\varepsilon)$, independently of $n$ and of
    the target blocks.
    Combining this with the deterministic construction bound
    gives an absolute constant $C_0>0$ such that
    \begin{equation}
        \mathcal{T}^{\mathrm{con}}_\varepsilon(U)
        \leq m+C_0\sqrt{2^n\log_2(1/\varepsilon)}.
    \end{equation}
    
    Let $s$ be the even threshold in the theorem.
    On the event $m\leq s$, we therefore have
    \begin{equation}
        \mathcal{T}^{\mathrm{con}}_\varepsilon(U)
        \leq s+C_0\sqrt{2^n\log_2(1/\varepsilon)}.
    \end{equation}
    Equation~\eqref{eq:count_upper} and a union bound imply
    \begin{align*}
        &\Pr\!\left[
            \mathcal{T}^{\mathrm{con}}_\varepsilon(U)
            >s+C_0\sqrt{2^n\log_2(1/\varepsilon)}
        \right]\\
        &\qquad\leq\Pr[m>s]\\
        &\qquad\leq C\,2^n
            \frac{(s+1)^2}{2^s\varepsilon^3}\\
        &\qquad\leq
            C_{\delta,\iota}\,
            2^{n(1-\iota)}\varepsilon^\delta(\log_2(1/\varepsilon)+n+1)^2.
    \end{align*}
    This proves Equation~\eqref{eq:thm1_full}.

    Finally, fix $n$ and $\delta>0$. Apply the preceding upper bound with a smaller positive slack, for example $\delta/2$, and fix $\iota=2$.
    Then $s=(3+\delta/2)\log_2(1/\varepsilon)+O_n(1)$ and
    $\sqrt{2^n\log_2(1/\varepsilon)}
    =O_n(\sqrt{\log_2(1/\varepsilon)})
    =o(\log_2(1/\varepsilon))$. For sufficiently small
    $\varepsilon$, the construction consequently has T-count
    at most $(3+\delta)\log_2(1/\varepsilon)$, except on an event of probability
    $O_{n,\delta}(\varepsilon^{\delta/2}(\log_2(1/\varepsilon)+1)^2)$.
    Its dirty ancillae may all be initialized to zero;
    because they are restored exactly, the resulting circuit
    is admissible in the clean-ancilla model. Thus
    \begin{equation*}
        \mathcal{T}^{\mathrm{clean}}_\varepsilon(U)
        \leq\mathcal{T}^{\mathrm{con}}_\varepsilon(U).
    \end{equation*}
    Equation~\eqref{eq:clean_ancilla_asymptotic} bounds the
    probability that the clean-ancilla minimum is below
    $(3-\delta)\log_2(1/\varepsilon)$ by $\varepsilon^{\delta-o(1)}$ when
    $0<\delta<3$; for $\delta\geq3$ the required lower bound
    is automatic. Absorbing the logarithmic factor into a
    smaller positive power of $\varepsilon$ and combining
    the two failure probabilities proves
    Equation~\eqref{eq:thm1_fixed_n}.
\end{proof}

\subsection{Ancilla-free controlled gate synthesis}\label{ssec:thm2}
In Section \ref{ssec:construction}, controlled gates $P_j\oplus Q_j$ in Equation (\ref{eq:pq}) take the qubits in $\mathcal{C}$ as control qubits and the qubit in $\mathcal{B}$ as a target qubit.
However, if we limit ourselves to the 2-qubit case (i.e., $n=1$), we can construct an equivalent Clifford+T circuit without ancilla.
\begin{rthm2} [$\mathrm{SU}(2)\oplus\mathrm{SU(2)}$ ancilla-free controlled gate synthesis]
    Let $U=U_1\oplus U_2$ with $U_1,U_2\in\mathrm{SU(2)}$, and define
    \begin{equation*}
        t_{\max}:=\max\left\{\mathcal{T}_\varepsilon(U_1),\mathcal{T}_\varepsilon(U_2)\right\}.
    \end{equation*}
    The gate $U$ can be implemented up to error $\varepsilon$ by an ancilla-free Clifford+T circuit using at most
    \begin{equation} \label{eq:thm2_instance}
        t_{\max}+O(1)
    \end{equation}
    T gates.
    For independent Haar-random target blocks, this T-count is
    \begin{equation} \label{eq:thm2_full}
        3\log_2(1/\varepsilon) + o\left(\log(1/\varepsilon)\right)
        \quad \mathrm{w.h.p.}
    \end{equation}
    Furthermore, leading coefficient $3$ remains optimal even if arbitrary
    clean ancillae and non-block-diagonal approximating unitaries
    are allowed.
\end{rthm2}
Figure~\ref{fig:su2su2-ancilla-free} summarizes the ancilla-free circuit construction.
\begin{figure}[t]
    \centering
    \begin{tikzpicture}[
        x=1cm,
        y=1cm,
        font=\footnotesize,
        wire/.style={-{Latex[length=2mm]}, thick},
        endpoint/.style={draw, rounded corners, fill=blue!8,
            text width=1.45cm, minimum height=1.85cm,
            inner xsep=3pt, inner ysep=5pt, align=center},
        layers/.style={draw, rounded corners, fill=orange!13,
            text width=5.7cm, minimum height=2.35cm,
            inner xsep=3pt, inner ysep=5pt, align=center}
    ]
        \draw[wire] (0,1.15) -- (12.7,1.15);
        \draw[wire] (0,0) -- (12.7,0);
        \node[left] at (0,1.15) {$|i\rangle_{\mathcal A}$};
        \node[left] at (0,0) {$|\psi\rangle_{\mathcal B}$};

        \node[endpoint] at (1.0,0.575) {$C_1'\oplus C_2'$};
        \node[layers] at (5.3,0.575) {
            aligned normal-form layers\\[4pt]
            {\scriptsize $\displaystyle
            \begin{gathered}
                \prod_{j=1}^{m}
                \bigl((M_{j,1}'\oplus M_{j,2}')(H_{j,1}\oplus H_{j,2})\bigr)
            \end{gathered}$}
        };
        \node[endpoint] at (9.35,0.575) {$F_1\oplus F_2$};
        \node[endpoint, text width=1.55cm] at (11.5,0.575)
            {$\zeta^{s_1}I_2$\\$\oplus$\\$\zeta^{s_2}I_2$};

        \node[right] at (12.7,1.15) {$|i\rangle_{\mathcal A}$};
        \node[right] at (12.7,0) {$U_i'|\psi\rangle_{\mathcal B}$};
    \end{tikzpicture}
    \caption{Circuit structure for the ancilla-free $\mathrm{SU}(2)\oplus\mathrm{SU}(2)$ synthesis in Theorem~\ref{thm:su2su2}. The single qubit in $\mathcal A$ directly selects the two diagonal blocks in every layer, replacing the Boolean and instruction registers used for multiple controls. Among the aligned $H_{j,1}\oplus H_{j,2}$ layers, at most two are nontrivial controlled-H gates.}
    \label{fig:su2su2-ancilla-free}
\end{figure}
\begin{proof} [Proof of Theorem \ref{thm:su2su2}.]
    {\SY Since $\mathcal{A}$ is a single-qubit register when $n=1$, multi-bit control at Section \ref{ssec:construction} is just single-bit control. Thus,}
    by directly using the qubit in $\mathcal{A}$ as the control qubit instead of controlling via qubits in $\mathcal{C}$, the need for the Boolean layers is eliminated, allowing an ancilla-free Clifford+T circuit to be constructed for the approximation.
    For each $j$, the block $M_{j,1}^\prime\oplus M_{j,2}^\prime$ can be implemented with at most one $T$ gate.
    If its two blocks are equal, the common single-qubit operation is applied directly; otherwise, the block is equivalent up to Clifford conjugation to the gate $TSH\oplus\zeta TH$ in Equation (\ref{eq:tshth}).
    Hence all $m=t_{\max}$ blocks involving the $M_{j,i}^\prime$ require at most $t_{\max}$ $T$ gates in total.
    Moreover, each sequence contains exactly one inserted $H$ gate, namely at $j=(m_i+m)/2$.
    Therefore, among all blocks $H_{j,1}\oplus H_{j,2}$, at most two are nontrivial controlled-$H$ gates.
    Since a controlled-$H$ gate has T-count $2$, their total contribution is at most four $T$ gates.
    The blocks involving $F_i$, $C_i^\prime$, and $\zeta^{s_i}I_2$ require only $O(1)$ additional $T$ gates because they are chosen from finite sets.
    This proves Equation (\ref{eq:thm2_instance}).
    For independent Haar-random $U_1,U_2$, applying Lemma \ref{lmm:bound} with $k=2$ yields Equation (\ref{eq:thm2_full}).
    Taking $n=1$ in
    Lemma~\ref{lmm:clean_ancilla_lower_bound}
    gives a matching lower bound, even when arbitrary clean
    ancillae are allowed.
    Thus, the leading coefficient $3$ is optimal.
\end{proof}


\section{Applications}\label{sec:applications}
The constructions below give deterministic bounds in terms of the
individual synthesis costs. We then justify separately the probability
laws needed for their leading-coefficient statements.

The uniform upper bound below follows from the covering estimates of
Parzanchevski and Sarnak~\cite[Corollary~3.2 and Equation~(5.3)]{PARZANCHEVSKI2018869}.
We record it in the half-diamond-distance and even-T-count conventions
of this paper, and include a short rederivation from
Equation~\eqref{eq:even_t_upper_tail}.

\begin{lmm}[A uniform single-qubit bound]\label{lmm:uniform_su2_upper}
    There is a function
    \begin{equation}
        M_\eta=6\log_2(1/\eta)
        +O\!\left(\log\log(4/\eta)\right)
        \qquad (0<\eta<1)
    \end{equation}
    such that $\mathcal T_\eta(W)\leq M_\eta$ for every
    $W\in\mathrm{SU}(2)$.
\end{lmm}
\begin{proof}
    Let $\mu$ be normalized Haar measure on $\mathrm{SU}(2)$.
    Equation~\eqref{eq:even_t_upper_tail} implies that, for an even
    integer $s=3\log_2(1/\eta)+O(\log\log(4/\eta))$ chosen
    sufficiently large, the set
    \[
        E=\{V:\mathcal T_{\eta/2}(V)\leq s\}
    \]
    has measure greater than $1/2$.
    For any fixed $W$, Haar invariance gives
    $\mu(WE^\dagger)=\mu(E)>1/2$, so
    $E\cap WE^\dagger$ is nonempty.
    Consequently $W=V_1V_2$ for some $V_1,V_2\in E$.
    Composing their approximating circuits gives error at most
    $\eta$ and even T-count at most $2s$.
    This is an existence bound and makes no claim about efficient
    classical synthesis.
\end{proof}

\begin{lmm}[A rotation-angle upper tail]
\label{lmm:rotation_upper_tail}
    Write $R_z(\theta)=\operatorname{diag}
    (e^{-i\theta/2},e^{i\theta/2})$, and let
    \begin{equation}\label{eq:rotation_bad_measure}
        p_\eta(s)=\frac{1}{2\pi}\operatorname{meas}
        \bigl\{\beta\in[0,2\pi):
        \mathcal T_\eta(R_z(2\beta))>s\bigr\}.
    \end{equation}
    Thus $p_\eta(s)$ is the fraction of diagonal phases whose
    synthesis requires more than $s$ T gates. There is an absolute
    constant $C>0$ such that, for $0<\eta\leq1/2$ and even $s\geq2$,
    \begin{equation}\label{eq:rotation_upper_tail}
        p_\eta(s)\leq
        C\frac{(s+1)^2}{2^s\eta^5}.
    \end{equation}
\end{lmm}
\begin{proof}
    A Haar-random $W\in\mathrm{SU}(2)$ can be written as
    \begin{equation*}
        W=\begin{pmatrix}
            \sqrt{1-t}\,e^{-i\beta}&-\sqrt t\,e^{-i\psi}\\
            \sqrt t\,e^{i\psi}&\sqrt{1-t}\,e^{i\beta}
        \end{pmatrix},
    \end{equation*}
    where $t$ is uniform on $[0,1]$ and $\beta,\psi$ are independent
    uniform angles on $[0,2\pi)$, independent of $t$.
    Indeed, this parametrizes normalized surface measure on the
    unit sphere in $\mathbb C^2$, which is Haar measure on
    $\mathrm{SU}(2)$. Since
    $\operatorname{tr}(R_z(2\beta)^\dagger W)=2\sqrt{1-t}$,
    \begin{equation*}
        D_\diamond(R_z(2\beta),W)=\sqrt t.
    \end{equation*}
    If $\mathcal T_\eta(R_z(2\beta))>s$ and $t<\eta^2/4$, then
    $\mathcal T_{\eta/2}(W)>s$: otherwise the same even-T-count
    circuit would approximate $R_z(2\beta)$ within $\eta$ by
    the triangle inequality. Consequently,
    \begin{equation*}
        \frac{\eta^2}{4}p_\eta(s)
        \leq\Pr[\mathcal T_{\eta/2}(W)>s]
        \leq C\frac{(s+1)^2}{2^s(\eta/2)^3},
    \end{equation*}
    where the last inequality is Equation~\eqref{eq:count_upper}.
    Absorbing the constant factor proves the claim.
\end{proof}
The factor $\eta^2$ is the Haar volume in the two directions
transverse to the rotation subgroup. This argument therefore
gives the coefficient $5$ for uniformly distributed rotation
angles; it does not give the conjectured coefficient $3$.

\begin{rthm3}[Controlled gate synthesis]
    Let $n\geq1$ and $U=\bigoplus_{j=1}^{2^n}U_j$, where
    $U_j\in\mathrm{U}(2)$.
    The decomposition below defines determinant-one blocks $V_j$
    and rotations $D_j$ (including $D_{2^n}=I_2$).
    Put
    \[
        a=\max_j\mathcal T_{\varepsilon/2}(V_j),
        \qquad r=\max_j\mathcal T_{\varepsilon/2}(D_j).
    \]
    There is a Clifford+T circuit with clean ancillae and error
    at most $\varepsilon$ whose T-count $N_\varepsilon(U)$ satisfies
    \begin{equation}\label{eq:thm3_instance}
        N_\varepsilon(U)\leq
        a+r+O\!\left(\sqrt{2^n(a+1)}+\sqrt{2^n(r+1)}+n\right).
    \end{equation}
    For fixed $n$ and independent Haar-random $U_j\in\mathrm{U}(2)$,
    this gives the unconditional upper bound
    \begin{equation}\label{eq:thm3_full}
        8\log_2(1/\varepsilon)
        +O\!\left(\sqrt{2^n\log(1/\varepsilon)}\right)
        +o\!\left(\log(1/\varepsilon)\right)
        \quad\mathrm{w.h.p.}
    \end{equation}
    In addition, suppose that
    \begin{equation}\label{eq:rz_generic_hypothesis}
        \limsup_{\eta\to0}
        \frac{\mathcal T_\eta(R_z(2\beta))}{\log_2(1/\eta)}
        \leq3
        \quad\text{for Lebesgue-almost every }\beta.
    \end{equation}
    Then, for every fixed $\delta>0$,
    \begin{equation}\label{eq:thm3_conditional_six}
        \Pr\!\left[N_\varepsilon(U)
            \leq(6+\delta)\log_2(1/\varepsilon)\right]
        \longrightarrow1.
    \end{equation}
    If, more strongly, for every $\gamma>0$ there are
    $C_\gamma,c_\gamma>0$ such that
    \begin{equation}\label{eq:rz_quantitative_hypothesis}
        p_\eta\!\left(
            2\left\lceil\frac{(3+\gamma)\log_2(1/\eta)}{2}
            \right\rceil\right)
        \leq C_\gamma\eta^{c_\gamma},
    \end{equation}
    for all sufficiently small $\eta$, then the coefficient $8$
    in Equation~\eqref{eq:thm3_full} can be replaced by $6$
    with the same w.h.p. convention as in Section~\ref{ssec:asymp}.
\end{rthm3}
\begin{proof}[Proof of Theorem~\ref{thm:general_control}]
    Choose $\alpha_j\in[0,2\pi)$ with
    $\det U_j=e^{i\alpha_j}$, and set
    \begin{equation*}
        A=\sum_{k=1}^{2^n}\alpha_k,\qquad
        \widetilde U=e^{-iA/(2^{n+1})}U,\qquad
        \phi_j=\alpha_j-A/2^n.
    \end{equation*}
    This global phase does not change the channel, and
    $\sum_j\phi_j=0$, so $\widetilde U\in\mathrm{SU}(2^{n+1})$.
    Define
    \[
        \theta_j=2\sum_{k=1}^j\phi_k,
        \qquad
        D=\left(\bigoplus_{j=1}^{2^n-1}R_z(-\theta_j)\right)
          \oplus I_2.
    \]
    Let $P=\mathrm{INC}_{n+1}$, where
    $P|x\rangle=|x+1\bmod 2^{n+1}\rangle$.
    The diagonal matrix $Q=PDP^\dagger=\bigoplus_jQ_j$ satisfies
    $\det Q_j=e^{i\phi_j}$, including $j=2^n$ because
    $\prod_j e^{i\phi_j}=1$.
    Hence $V_j=\widetilde U_jQ_j^\dagger\in\mathrm{SU}(2)$ and
    \begin{equation}\label{eq:general_control_decomposition}
        \widetilde U=\left(\bigoplus_jV_j\right)PDP^\dagger.
    \end{equation}
    Apply the instance bound of Theorem~\ref{thm:su2cnt} to
    the two direct sums, with error $\varepsilon/2$ each.
    The increment and its inverse require $O(n)$ T gates and
    clean scratch qubits~\cite{khattar2025rise}.
    The triangle inequality proves Equation~\eqref{eq:thm3_instance}.

    For the probabilistic statement, condition on the determinants
    of the input blocks. Each $V_j$ is then a fixed translate of
    an independent Haar-random determinant-one factor, and so the
    $V_j$ remain independent Haar-random elements of $\mathrm{SU}(2)$.
    Thus Lemma~\ref{lmm:bound} gives
    $a\leq 3\log_2(1/\varepsilon)+o(\log(1/\varepsilon))$
    w.h.p. for fixed $n$.
    To control $r$, write $D_j=R_z(2\beta_j)$ for $j<2^n$, where
    \begin{equation}\label{eq:determinant_rotation_phases}
        \beta_j=-\left(1-\frac{j}{2^n}\right)\sum_{k\leq j}\alpha_k
                +\frac{j}{2^n}\sum_{k>j}\alpha_k.
    \end{equation}
    The determinant arguments $\alpha_k$ are independent and
    uniform on $[0,2\pi)$. For each $j<2^n$, one of their
    coefficients in this expression has absolute value
    $q=\max\{j/2^n,1-j/2^n\}\in[1/2,1)$.
    Conditional on all other arguments, that term modulo $2\pi$
    is uniform on an arc of length $2\pi q$; its density relative
    to the uniform circle measure is at most $1/q\leq2$.
    Averaging the conditioning preserves this bound. Thus each
    $\beta_j\bmod2\pi$ has density at most $2$, although the
    $\beta_j$ need not be independent.
    Since $\mathcal T_\eta(I_2)=0$, a union bound gives
    \begin{equation}\label{eq:thm3_rotation_transfer}
        \Pr[r>s]\leq2(2^n-1)p_{\varepsilon/2}(s).
    \end{equation}
    Lemma~\ref{lmm:rotation_upper_tail} now implies, for fixed $n$,
    $r\leq5\log_2(1/\varepsilon)+o(\log(1/\varepsilon))$
    w.h.p. More explicitly, putting $\eta=\varepsilon/2$ and
    \begin{equation*}
        s_k=2\left\lceil
            \frac{(k+\gamma)\log_2(1/\eta)}{2}\right\rceil
        \qquad(k=3,5),
    \end{equation*}
    the probability that $a>s_3$ or $r>s_5$ is at most
    $C_\gamma 2^n\,\eta^\gamma(\log_2(1/\eta)+1)^2$.
    Substitution in Equation~\eqref{eq:thm3_instance} proves
    Equation~\eqref{eq:thm3_full}.

    Under Equation~\eqref{eq:rz_generic_hypothesis}, bounded
    convergence gives $p_\eta(s_3)\to0$ for every $\gamma>0$.
    Equation~\eqref{eq:thm3_rotation_transfer} therefore gives
    $r\leq s_3$ with probability tending to one. Combining this
    with the bound for $a$ and choosing $2\gamma<\delta$ proves
    Equation~\eqref{eq:thm3_conditional_six}; the square-root
    terms are $o(\log(1/\varepsilon))$ for fixed $n$.
    Under Equation~\eqref{eq:rz_quantitative_hypothesis}, the
    same argument gives a polynomially small failure probability.

    Finally, when $m>0$, the identity block has $m_{2^n}=0$ and inserts its $H$
    at the single position $m/2$. Equation~\eqref{eq:occupied_h_layers}
    counts only occupied positions, with
    $K\leq\min\{2^n,m/2+1\}\leq\sqrt{2^n(m+1)}$.
    Hence this block causes no additional leading-order cost,
    regardless of the minimum synthesis cost among the blocks.
    The case $m=0$ uses the Clifford template described in
    Section~\ref{ssec:construction}.
\end{proof}

\paragraph{Relation to Ross--Selinger synthesis.}
Equation~\eqref{eq:rz_generic_hypothesis} is implied by the
generic-angle part of the Ross--Selinger conjecture
\cite[Section~8.3, Conjecture~1, p.~922]{10.5555/3179330.3179331}:
the angles with $\tan(\theta/2)\in\mathbb Q(\sqrt2)$ form a
countable set. Their determinant-one approximants have even
T-count~\cite[Section~7.1, Lemma~10]{10.5555/3179330.3179331},
and operator-norm error bounds imply the diamond-distance
bounds used here. Thus neither T-count parity nor the fixed
identity block obstructs the coefficient $6$.
However, optimality of their synthesis algorithm for each input
does not prove the conjectured asymptotic cost of that input.
The generic-angle conjecture also does not supply the quantitative
angular tail in Equation~\eqref{eq:rz_quantitative_hypothesis}.
Accordingly, the unconditional coefficient $8$, the conditional
coefficient $6$ with probability tending to one, and its stronger
w.h.p. version are distinct statements.

\begin{lmm}[An exact-word synthesis subroutine]\label{lmm:su2su2su2}
    Let $W\in\mathrm{SU}(2)$, $0<\eta<1$, and
    $m=\mathcal T_\eta(W)$.
    There is an ancilla-free Clifford+T circuit implementing
    $A\oplus B$, with $A,B\in\mathrm{SU}(2)$, at T-count
    at most $m/2+O(1)$, such that
    \begin{equation}\label{eq:split_phase_aligned}
        D_\diamond(AB^\dagger,W)\leq\eta,
        \qquad
        \operatorname{Re}\operatorname{tr}(W^\dagger AB^\dagger)
            \geq0,
        \qquad
        \|AB^\dagger-W\|_{\mathrm{op}}\leq\sqrt2\,\eta.
    \end{equation}
    For a Haar-random $W$, this T-count is at most
    $1.5\log_2(1/\eta)+o(\log(1/\eta))$ w.h.p.
\end{lmm}
\begin{proof}
    Choose an even-T-count word approximating $W$, and split it
    at a gate boundary into words of lengths
    \[
        m_1=2\lfloor m/4\rfloor,
        \qquad m_2=m-m_1.
    \]
    Both lengths are even, $|m_1-m_2|\leq2$, and their maximum
    is at most $m/2+1$.
    Normalize the determinants of the two partial products by
    powers of $\zeta$; because their T-counts are even, these
    phases suffice to obtain $A,B\in\mathrm{SU}(2)$ whose
    product $AB^\dagger$ agrees with the original word up to
    a scalar phase. Apply the exact-word construction from the
    proof of Theorem~\ref{thm:su2su2} to these two partial
    products. This uses at most $\max(m_1,m_2)+O(1)$ T gates;
    no randomness assumption is made about the partial products.
    The case $m=0$ uses only the finite set of controlled
    Clifford blocks and has constant cost.

    Replace $B$ by $-B$ if necessary to obtain the trace condition
    in Equation~\eqref{eq:split_phase_aligned}. This is a $Z$
    gate on the control and adds no T gates.
    If the eigenvalues of $W^\dagger AB^\dagger$ are
    $e^{\pm i\alpha}$, the trace condition allows
    $|\alpha|\leq\pi/2$. Therefore
    \[
        \|AB^\dagger-W\|_{\mathrm{op}}
          =2|\sin(\alpha/2)|
          \leq\sqrt2|\sin\alpha|
          =\sqrt2 D_\diamond(AB^\dagger,W).
    \]
    Finally, the Haar-random assertion follows by applying
    Lemma~\ref{lmm:bound} to $W$ itself, rather than to the
    two partial products.
\end{proof}

\begin{rthm4}[Ancilla-free controlled gate synthesis]
    Let $U=\bigoplus_{x=0}^{2^n-1}U_x$ with
    $U_x\in\mathrm{SU}(2)$.
    The parity decomposition described below defines $2^n-1$
    successive single-qubit targets $W_1,\ldots,W_{2^n-1}$ and a
    final residual target $W_0$.
    With $\eta=\varepsilon/(2^{n+1})$, an ancilla-free Clifford+T
    circuit approximates $U$ within error $\varepsilon$ at
    T-count at most
    \begin{equation}\label{eq:thm4_instance}
        \frac12\sum_{i=1}^{2^n-1}\mathcal T_\eta(W_i)
        +\mathcal T_\eta(W_0)+O(2^n).
    \end{equation}
    In particular, an unconditional bound valid for every target is
    \begin{equation}\label{eq:thm4_uniform}
        3(2^n+1)\log_2(1/\varepsilon)
        +O\!\left(n2^n+2^n\log\log(2^{n+2}/\varepsilon)\right).
    \end{equation}
    For fixed $n$ and independent Haar-random $U_x$, a fixed
    measurable choice of the parity decomposition and of every
    synthesis subroutine gives
    \begin{equation}\label{eq:thm4_full}
        \frac{3(2^n+1)}2\log_2(1/\varepsilon)
        +O(n2^n)+o\!\left(2^n\log(1/\varepsilon)\right)
        \quad\mathrm{w.h.p.}
    \end{equation}
    No independence or Haar-distribution assumption is imposed
    on the successive targets $W_i$.
\end{rthm4}
\begin{proof}[Proof of Theorem~\ref{thm:su2noancilla}]
    Write $p_i(x)$ for the binary inner product of $i$ and $x$,
    taken modulo $2$, and define
    \[
        P_i(A,B)=\bigoplus_{x=0}^{2^n-1}
        \begin{cases}
            A,&p_i(x)=0,\\
            B,&p_i(x)=1.
        \end{cases}
    \]
    For $i\ne0$, this is the two-qubit gate $A\oplus B$
    conjugated by a CNOT circuit that computes the parity
    $p_i(x)$ on one existing control qubit and then uncomputes
    it. It therefore requires no ancillae and has the same
    T-count as $A\oplus B$.
    The exact parity decomposition of
    Appendix~A of~\cite{Yamazaki2025CliffordVSF} gives
    $R,\beta_1,\ldots,\beta_{2^n-1}\in\mathrm{SU}(2)$ such that
    \begin{equation}\label{eq:parity_exact_decomposition}
        U=(I_{2^n}\otimes R)
          \prod_{i=1}^{2^n-1}P_i(I_2,\beta_i).
    \end{equation}
    Only this algebraic decomposition, not any distributional
    assertion about its factors, is used below.

    Set $q=2^n-1$. Choose $A_i,B_i$ in descending order
    $i=q,q-1,\ldots,1$.
    Once the later pairs are fixed, put
    \[
        G_i=A_{i+1}\cdots A_q,
        \qquad W_i=G_i\beta_iG_i^\dagger.
    \]
    Apply Lemma~\ref{lmm:su2su2su2} to $W_i$ at accuracy $\eta$
    and take the adjoint of the resulting direct-sum circuit.
    Its blocks $A_i,B_i$ then satisfy
    \[
        \|A_i^\dagger B_i-W_i\|_{\mathrm{op}}
            \leq\sqrt2\,\eta
    \]
    at cost $\mathcal T_\eta(W_i)/2+O(1)$.
    Let $G=A_1\cdots A_q$, $W_0=RG^\dagger$, and choose an
    even-T-count approximation $C$ to $W_0$ at accuracy $\eta$.
    Choose its determinant-one representative and sign so that
    $\|C-W_0\|_{\mathrm{op}}\leq\sqrt2\,\eta$.
    Form the circuit
    \[
        \widetilde U=(I_{2^n}\otimes C)
              \prod_{i=1}^qP_i(A_i,B_i).
    \]
    The exact identity
    \begin{equation}
        \prod_{i=1}^qP_i(A_i,B_i)
        =(I_{2^n}\otimes G)
          \prod_{i=1}^q
          P_i\!\left(I_2,G_i^\dagger A_i^\dagger B_iG_i\right)
    \end{equation}
    shows why the adaptive conjugations $G_i$ are needed.
    Comparing this identity with
    Equation~\eqref{eq:parity_exact_decomposition}, the
    telescoping bound for products of unitaries gives
    \[
        D_\diamond(U,\widetilde U)
          \leq\|U-\widetilde U\|_{\mathrm{op}}
          \leq 2^n\sqrt2\,\eta<\varepsilon.
    \]
    Summing the subroutine costs proves
    Equation~\eqref{eq:thm4_instance}.
    Lemma~\ref{lmm:uniform_su2_upper} gives
    Equation~\eqref{eq:thm4_uniform} for arbitrary successive
    targets, even when they depend on earlier approximations.

    It remains to justify the probabilistic assertion, including
    uniformity in the chosen approximants.
    Write $M=\mathrm{SU}(2)^{2^n}$, with product Haar measure $\mu$,
    and consider the forward parity map
    \begin{equation}\label{eq:parity_forward_map}
        \Phi(R,\beta_1,\ldots,\beta_q)
        =\left(R\prod_{i=1}^q\beta_i^{p_i(x)}\right)_{x=0}^{2^n-1}.
    \end{equation}
    The exact decomposition makes $\Phi$ surjective.
    Fix measurable square-root branches in that decomposition,
    thereby obtaining a measurable right inverse $S$ of $\Phi$.
    The domain and range have the same real dimension $3\cdot2^n$.
    Equip them with the product Frobenius metric. Telescoping
    each product in Equation~\eqref{eq:parity_forward_map} shows
    that $\Phi$ is $2^n$-Lipschitz. A Lipschitz map increases
    $3\cdot2^n$-dimensional Hausdorff measure by at most the $3\cdot2^n$-th
    power of its Lipschitz constant, directly by mapping the
    sets in a Hausdorff-measure covering. Normalized Hausdorff
    measure on $M$ equals product Haar measure. Hence
    \[
        \mu(\Phi(B))\leq 2^{3n2^n}\mu(B)
    \]
    for every measurable $B\subseteq M$.
    Since $\Phi\circ S$ is the identity, Haar-random input gives
    a factor distribution $\nu$ satisfying
    \begin{equation}\label{eq:parity_factor_density}
        \nu(B)=\mu\{U:S(U)\in B\}
        \leq\mu(\Phi(B))\leq 2^{3n2^n}\mu(B).
    \end{equation}

    At each accuracy $\eta$, choose the output of the subroutine
    as a fixed measurable function of its current target.
    The transformation
    \[
        (R,\beta_1,\ldots,\beta_q)
           \longmapsto(W_0,W_1,\ldots,W_q)
    \]
    preserves product Haar measure: in descending order, $G_i$
    depends only on the later $\beta_j$, so the change
    $\beta_i\mapsto G_i\beta_iG_i^\dagger$ is a fixed
    Haar-preserving conjugation conditional on those later
    coordinates. Finally, $R\mapsto RG^\dagger$ is a
    Haar-preserving translation conditional on all $\beta_i$.
    Repeated integration in these coordinates proves the claim;
    it does not require the synthesis choices to be continuous.
    Applying this transformation to
    Equation~\eqref{eq:parity_factor_density} shows that the
    joint law of the actual targets $W_i$ is bounded by
    $2^{3n2^n}\mu$, uniformly in $\eta$.
    A union bound and Equation~\eqref{eq:even_t_upper_tail}
    therefore give, for every even $s\ge2$,
    \begin{equation}\label{eq:recursive_synthesis_tail}
        \Pr\!\left[
          \max_{0\leq i<2^n}\mathcal T_\eta(W_i)>s
        \right]
        \leq C 2^{n(3\cdot2^n+1)}\frac{s^2}{2^s\eta^3}.
    \end{equation}
    For fixed $n$, take the least even integer
    $s\geq(3+\delta)\log_2(1/\eta)$ and absorb the squared
    logarithm into a smaller positive power of $\eta$.
    Substitution into Equation~\eqref{eq:thm4_instance}, followed
    by $\log_2(1/\eta)=\log_2(1/\varepsilon)+n+1$, proves
    Equation~\eqref{eq:thm4_full} for every $\delta>0$.
\end{proof}

Third, the controlled synthesis subroutine also applies to general
two-qubit gates.  Here the probability measure is Haar measure on
$\mathrm{SU}(4)$.  The distributions of the intermediate targets,
including those selected after an earlier approximation, are justified
in Appendix~\ref{sec:SU(4)synthesis}.
\begin{rthm5}[SU(4) gate synthesis]
    Let $U$ be Haar-random in $\mathrm{SU}(4)$.
    There is an ancilla-free Clifford+T construction with error at most $\varepsilon$ and T-count $N_\varepsilon(U)$ such that, for every $\gamma>0$, there are constants $C_\gamma,c_\gamma>0$ for which
    \begin{equation}\label{eq:thm5_full}
        \Pr\!\left[N_\varepsilon(U)>(9+\gamma)\log_2(1/\varepsilon)\right]
        \leq C_\gamma\varepsilon^{c_\gamma}
    \end{equation}
    for all sufficiently small $\varepsilon>0$.
\end{rthm5}

\begin{proof}[Proof of Theorem~\ref{thm:su4synthesis}]
    Put $\delta=\varepsilon/4$.  We fix measurable choices in every
    decomposition and synthesis subroutine; for example, a fixed ordering
    of finite gate words resolves ties between minimum-cost approximants.
    The covering $\mathrm{SU}(4)\to\mathrm{SO}(6)$ described in
    Appendix~\ref{sec:su4so6} gives the decomposition
    \begin{equation}\label{eq:su4_first_decomposition}
        \begin{aligned}
            U={}&(e^{i\theta_1X}e^{i\theta_2Y}\otimes I_2)
            P_1U_1,\\
            P_1={}&e^{i\theta_3Z\otimes X}
            e^{i\theta_4Z\otimes Y}e^{i\theta_5Z\otimes Z},
        \end{aligned}
    \end{equation}
    where $U_1$ belongs to the copy $K_5\cong\mathrm{Spin}(5)$
    fixing the sixth coordinate in the SO(6) representation.
    Write
    \begin{equation*}
        \begin{gathered}
            A_1=e^{i\theta_3X}e^{i\theta_4Y}e^{i\theta_5Z},\qquad
            B_1=e^{-i\theta_3X}e^{-i\theta_4Y}e^{-i\theta_5Z},\\
            G_1=A_1B_1^\dagger,\qquad P_1=A_1\oplus B_1.
        \end{gathered}
    \end{equation*}
    Lemma~\ref{lmm:su2su2su2}, applied at precision $\delta$,
    constructs $Q_1=V_1\oplus V_2$ with
    $D_\diamond(V_1V_2^\dagger,G_1)\leq\delta$ and cost at most
    $\frac12\mathcal T_\delta(G_1)+O(1)$.
    Choose the sign of $V_2$ so that
    \begin{equation}\label{eq:su4_first_sign}
        \operatorname{Re}\operatorname{tr}
        (G_1^\dagger V_1V_2^\dagger)\geq0.
    \end{equation}
    Multiplication of $Q_1$ by $Z\otimes I_2$ makes this
    sign change without adding a T gate.

    The correction $K_1=I_2\otimes(V_1^\dagger A_1)$ lies in
    $K_5$.  Decompose the corrected residual as
    \begin{equation}\label{eq:su4_second_decomposition}
        \begin{aligned}
            K_1U_1={}&(e^{i\theta_6Z}\otimes I_2)P_2U_2,\\
            P_2={}&e^{i\theta_7Y\otimes X}
            e^{i\theta_8Y\otimes Y}e^{i\theta_9Y\otimes Z},
        \end{aligned}
    \end{equation}
    where $U_2$ belongs to $K_4\cong\mathrm{Spin}(4)$,
    the subgroup fixing the fifth and sixth coordinates.
    Define
    \begin{equation*}
        \begin{gathered}
            A_2=e^{i\theta_7X}e^{i\theta_8Y}e^{i\theta_9Z},\qquad
            B_2=e^{-i\theta_7X}e^{-i\theta_8Y}e^{-i\theta_9Z},\\
            G_2=A_2B_2^\dagger,\qquad
            C=e^{i\theta_1X}e^{i\theta_2Y}e^{i\theta_6Z}.
        \end{gathered}
    \end{equation*}
    Synthesize $V_3$ approximating $C$ within $\delta$, with
    cost $\mathcal T_\delta(C)$.  Apply
    Lemma~\ref{lmm:su2su2su2} to $G_2$ to construct
    $Q_2=V_4\oplus V_5$ with cost at most
    $\frac12\mathcal T_\delta(G_2)+O(1)$, and choose the sign
    of $V_5$ so that
    \begin{equation}\label{eq:su4_second_sign}
        \operatorname{Re}\operatorname{tr}
        (G_2^\dagger V_4V_5^\dagger)\geq0.
    \end{equation}
    Again the sign correction is Clifford.

    Finally, put $K_2=I_2\otimes(V_4^\dagger A_2)$ and write
    \begin{equation}\label{eq:su4_final_blocks}
        A\oplus B=K_2(H\otimes I_2)U_2(H\otimes I_2),
        \qquad A,B\in\mathrm{SU}(2).
    \end{equation}
    The inclusion follows from the Pauli generators of $K_4$;
    conjugating them by $H\otimes I_2$ gives the generators of
    $\mathrm{SU}(2)\oplus\mathrm{SU}(2)$.
    The instance-dependent part of Theorem~\ref{thm:su2su2}
    approximates this direct sum within $\delta$ by
    $Q_3=V_6\oplus V_7$ at cost at most
    $\max\{\mathcal T_\delta(A),\mathcal T_\delta(B)\}+O(1)$.
    The resulting circuit is
    \begin{equation}\label{eq:su4_output_circuit}
        \begin{aligned}
            U'={}&(V_3\otimes I_2)Q_1(SH\otimes I_2)Q_2\\
                &\quad\times(HS^\dagger H\otimes I_2)Q_3(H\otimes I_2).
        \end{aligned}
    \end{equation}
    Lemma~\ref{lmm:su4kinji} proves $D_\diamond(U,U')\leq\varepsilon$.
    Its T-count satisfies the deterministic, instance-dependent bound
    \begin{equation}\label{eq:su4_instance_cost}
        \begin{aligned}
            N_\varepsilon(U)\leq{}&
            \tfrac12\mathcal T_\delta(G_1)+\mathcal T_\delta(C)
            +\tfrac12\mathcal T_\delta(G_2)\\
            &+\max\{\mathcal T_\delta(A),\mathcal T_\delta(B)\}+O(1).
        \end{aligned}
    \end{equation}

    Lemma~\ref{lmm:su4_target_distributions} shows that all five
    targets $G_1,C,G_2,A,B$ have Haar densities bounded by constants
    independent of $\varepsilon$.  This includes the dependence
    of the later targets on the earlier approximants.
    Therefore the single-qubit tail bound in
    Appendix~\ref{sec:even_t}, followed by a union bound, gives,
    for every $\eta>0$,
    \begin{equation*}
        \Pr\!\left[
        \max_{W\in\{G_1,C,G_2,A,B\}}\mathcal T_\delta(W)
            >(3+\eta)\log_2(1/\delta)\right]
        =O_\eta\!\left(\delta^\eta\log^2(1/\delta)\right).
    \end{equation*}
    No independence among these five targets is needed.
    The coefficients in Equation~\eqref{eq:su4_instance_cost}
    sum to $3$.  Choosing $\eta$ sufficiently small relative to
    $\gamma$ proves Equation~\eqref{eq:thm5_full}.
\end{proof}

\section*{Acknowledgments}
SY is supported by JSPS KAKENHI Grant No.~JP24KJ0907 and the Forefront Physics and Mathematics Program to Drive Transformation (FoPM).
SA is partially supported by JST PRESTO Grant No. JPMJPR2111, JST Moonshot R\&D MILLENNIA Program (Grant No. JPMJMS2061), JPMXS0120319794, and CREST (Japan Science and Technology Agency) Grant No. JPMJCR2113.

\paragraph{Use of generative AI.}
GPT-6 Astra (OpenAI) was used as a research-assistance tool
during the development of this work. The model identified
errors in earlier drafts of the arguments for
Theorems~\ref{thm:su2noancilla} and~\ref{thm:su4synthesis}
and proposed corrections. It also suggested the strategy of
using pairs of opposite off-diagonal blocks and 
restricting the size of $\varepsilon$ to say that their Frobenius norms are the same, as expressed in
Equation~\eqref{eq:clean_block_cancellation}. This strategy
is a key ingredient in the lower-bound proofs establishing
the optimality of the leading T-count coefficient in
Theorems~\ref{thm:su2cnt} and~\ref{thm:su2su2}.
The authors take full responsibility for the correctness,
originality, attribution, and presentation of all results
and proofs in this paper.

\bibliographystyle{plainnat}
\bibliography{ref}

@book{Nielsen_Chuang_2010, place={Cambridge}, title={Quantum Computation and Quantum Information: 10th Anniversary Edition}, publisher={Cambridge University Press}, author={Nielsen, Michael A. and Chuang, Isaac L.}, year={2010}}

@article{Amy2012AMA,
  title={A Meet-in-the-Middle Algorithm for Fast Synthesis of Depth-Optimal Quantum Circuits},
  author={Matthew Amy and Dmitrii L. Maslov and Michele Mosca and Martin R{\"o}tteler},
  journal={IEEE Transactions on Computer-Aided Design of Integrated Circuits and Systems},
  year={2012},
  volume={32},
  pages={818-830}
}

@article{Low2018TradingTG,
  title={Trading {T} gates for dirty qubits in state preparation and unitary synthesis},
  author={Guang Hao Low and Vadym Kliuchnikov and Luke Schaeffer},
  journal={Quantum},
  year={2024},
  volume={8},
  pages={1375}
}

@misc{Morisaki:2025mwy,
    author = "Morisaki, Hayata and Sano, Kaoru and Akibue, Seiseki",
    title = "{Optimal ancilla-free Clifford+T synthesis for general single-qubit unitaries}",
    eprint = "2510.05816",
    note = {arXiv:2510.05816},
    archivePrefix = "arXiv",
    primaryClass = "quant-ph",
    month = "10",
    year = "2025"
}

@misc{Yamazaki2025CliffordVSF,
  title={Clifford+V synthesis for multi-qubit unitary gates},
  author={Soichiro Yamazaki and Seiseki Akibue},
    eprint = "2510.08312",
    note = {arXiv:2510.08312},
    archivePrefix = "arXiv",
    primaryClass = "quant-ph",
    month = "10",
    year = "2025"
}

@article{PARZANCHEVSKI2018869,
title = {Super-Golden-Gates for PU(2)},
journal = {Advances in Mathematics},
volume = {327},
pages = {869-901},
year = {2018},
note = {Special volume honoring David Kazhdan},
issn = {0001-8708},
doi = {https://doi.org/10.1016/j.aim.2017.06.022},
url = {https://www.sciencedirect.com/science/article/pii/S0001870817301640},
author = {Ori Parzanchevski and Peter Sarnak}
}

@article{PhysRevA.71.022316,
  title = {Universal quantum computation with ideal Clifford gates and noisy ancillas},
  author = {Bravyi, Sergey and Kitaev, Alexei},
  journal = {Phys. Rev. A},
  volume = {71},
  issue = {2},
  pages = {022316},
  numpages = {14},
  year = {2005},
  month = {Feb},
  publisher = {American Physical Society},
  doi = {10.1103/PhysRevA.71.022316},
  url = {https://link.aps.org/doi/10.1103/PhysRevA.71.022316}
}

@article{10.5555/3179330.3179331,
author = {Ross, Neil J. and Selinger, Peter},
title = {Optimal ancilla-free Clifford+T approximation of z-rotations},
url = {https://www.rintonpress.com/xxqic16/qic-16-1112/0901-0953.pdf},
year = {2016},
issue_date = {September 2016},
publisher = {Rinton Press, Incorporated},
address = {Paramus, NJ},
volume = {16},
number = {11–12},
issn = {1533-7146},
journal = {Quantum Info. Comput.},
month = sep,
pages = {901–953},
numpages = {53}
}

@article{PhysRevLett.118.010501,
  title = {Optimal Hamiltonian Simulation by Quantum Signal Processing},
  author = {Low, Guang Hao and Chuang, Isaac L.},
  journal = {Phys. Rev. Lett.},
  volume = {118},
  issue = {1},
  pages = {010501},
  numpages = {5},
  year = {2017},
  month = {Jan},
  publisher = {American Physical Society},
  doi = {10.1103/PhysRevLett.118.010501},
  url = {https://link.aps.org/doi/10.1103/PhysRevLett.118.010501}
}

@article{Low2019hamiltonian,
  doi = {10.22331/q-2019-07-12-163},
  url = {https://doi.org/10.22331/q-2019-07-12-163},
  title = {Hamiltonian {S}imulation by {Q}ubitization},
  author = {Low, Guang Hao and Chuang, Isaac L.},
  journal = {{Quantum}},
  issn = {2521-327X},
  publisher = {{Verein zur F{\"{o}}rderung des Open Access Publizierens in den Quantenwissenschaften}},
  volume = {3},
  pages = {163},
  month = jul,
  year = {2019}
}

@article{GossetKothariWu2026,
  author = {Gosset, David and Kothari, Robin and Wu, Kewen},
  title = {Quantum State Preparation with Optimal {T}-Count},
  journal = {Quantum},
  volume = {10},
  pages = {2168},
  year = {2026},
  doi = {10.22331/q-2026-07-22-2168},
  url = {https://doi.org/10.22331/q-2026-07-22-2168}
}

@article{khattar2025rise,
  title={Rise of conditionally clean ancillae for efficient quantum circuit constructions},
  author={Khattar, Tanuj and Gidney, Craig},
  journal={Quantum},
  volume={9},
  pages={1752},
  year={2025},
  publisher={Verein zur F{\"o}rderung des Open Access Publizierens in den Quantenwissenschaften}
}

@article{KhanLeiMilicevic2019,
author = {Rizwanur Khan and Ruoyun Lei and Djordje Milićević},
title = {{Effective moments of Dirichlet $L$-functions in Galois orbits}},
volume = {12},
journal = {Involve: A Journal of Mathematics},
number = {3},
publisher = {MSP},
pages = {475 -- 490},
year = {2019},
doi = {10.2140/involve.2019.12.475},
URL = {https://doi.org/10.2140/involve.2019.12.475}
}

@article{Cappellini_2009,
doi = {10.1088/1751-8113/42/36/365209},
url = {https://doi.org/10.1088/1751-8113/42/36/365209},
year = {2009},
month = {aug},
publisher = {},
volume = {42},
number = {36},
pages = {365209},
author = {Cappellini, Valerio and Sommers, Hans-Jürgen and Bruzda, Wojciech and Życzkowski, Karol},
title = {Random bistochastic matrices},
journal = {Journal of Physics A: Mathematical and Theoretical}
}

@article{Stanley2016Smith,
  title={Smith normal form in combinatorics},
  author={Stanley, Richard P},
  journal={Journal of Combinatorial Theory, Series A},
  volume={144},
  pages={476--495},
  year={2016},
  publisher={Elsevier}
}

@misc{MilneANT2020,
author={Milne, James S.},
title={Algebraic Number Theory (v3.08)},
year={2020},
note={Available at www.jmilne.org/math/},
pages={166}
}

@misc{GilesSelinger2013NormalForm,
  author = {Giles, Brett and Selinger, Peter},
  title = {Remarks on {Matsumoto} and {Amano}'s normal form for single-qubit {Clifford+T} operators},
  year = {2013},
  note = {arXiv:1312.6584},
  url = {https://arxiv.org/abs/1312.6584}
}

@article{VatanWilliams2004,
  title = {Optimal quantum circuits for general two-qubit gates},
  author = {Vatan, Farrokh and Williams, Colin},
  journal = {Phys. Rev. A},
  volume = {69},
  issue = {3},
  pages = {032315},
  numpages = {5},
  year = {2004},
  month = {Mar},
  publisher = {American Physical Society},
  doi = {10.1103/PhysRevA.69.032315},
  url = {https://link.aps.org/doi/10.1103/PhysRevA.69.032315}
}

@article{WattsOConnorVala2013,
    author = "Watts, Paul and O'Connor, Maurice and Vala, Jiri",
    title = "{Metric Structure of the Space of Two-Qubit Gates, Perfect Entanglers and Quantum Control}",
    eprint = "1306.2811",
    archivePrefix = "arXiv",
    primaryClass = "quant-ph",
    doi = "10.3390/e15061963",
    journal = "Entropy",
    volume = "15",
    pages = "1963--1984",
    year = "2013"
}

\appendix
\section{Even-number T-count SU(2) synthesis} \label{sec:even_t}
Morisaki et al.~\cite{Morisaki:2025mwy} establish the leading
T-count coefficient $3$ for a Haar-random single-qubit target.
We record uniform tail estimates that also impose the even
T-count restriction used by our construction.
\begin{lmm}[Even-number T-count SU(2) gate synthesis]
\label{lmm:even_tcount_su2}
    Let $U$ be Haar-random in $\mathrm{SU}(2)$ and
    $0<\varepsilon\leq1/2$. There is an absolute constant
    $C>0$ such that, for every even integer $s\geq2$ and
    every real $r\geq0$,
    \begin{align*}
        \Pr[\mathcal{T}_\varepsilon(U)>s]
        &\leq C\frac{(s+1)^2}{2^s\varepsilon^3},\\
        \Pr[\mathcal{T}_\varepsilon(U)\leq r]
        &\leq C2^r\varepsilon^3.
    \end{align*}
    Consequently,
    \begin{equation} \label{eq:single_synthesis_morisaki}
        \mathcal{T}_\varepsilon(U)
        =3\log_2(1/\varepsilon)+o(\log(1/\varepsilon))
        \qquad\mathrm{w.h.p.}
    \end{equation}
\end{lmm}

\begin{proof} [Proof of Lemma \ref{lmm:even_tcount_su2}]
Although the target gates are stated as elements of $\mathrm{SU}(2)$,
we carry out the following argument in the projective unitary group
\begin{equation}
    \mathrm{PU}(2)
    =\mathrm{U}(2)/\mathrm{U}(1)
    \simeq \mathrm{SU}(2)/\{\pm I_2\}.
\end{equation}

We use the minimal Clifford+$\tau$ super-golden gate \cite{PARZANCHEVSKI2018869}. 
Let
\begin{equation}
    C_3
    =
    \left\{
        I_2,\,
        SH,\,
        HS^\dagger
    \right\}
    \subset \mathrm{Clifford}_1
\end{equation}
and
\begin{equation}
    \tau=TX.
\end{equation}
Note that every occurrence of $\tau$ can be implemented using
one ordinary $T$ gate and Clifford gates.

Let
\begin{equation}
    \Gamma=\langle C_3,\tau\rangle.
\end{equation}
For $s\geq1$, let $\Sigma_s$ denote the set of reduced words in
$\Gamma$ containing exactly $s$ occurrences of $\tau$:
\begin{equation}
    \Sigma_s
    =
    \left\{
        c_0\tau c_1\tau\cdots\tau c_s
        \,\middle|\,
        c_j\in C_3,\ 
        c_j\neq[I_2]\text{ for }1\leq j\leq s-1
    \right\}.
\end{equation}
By Proposition 4.1 of \cite{PARZANCHEVSKI2018869}, these reduced
words represent distinct elements of $\mathrm{PU}(2)$ and
\begin{equation}
    |\Sigma_s|
    =
    |C_3|^2(|C_3|-1)^{s-1}
    =
    9\cdot 2^{s-1}.
\end{equation}
Define the normalized averaging operator
$A_s:L^2(\mathrm{PU}(2))\rightarrow L^2(\mathrm{PU}(2))$ by
\begin{equation}
    (A_sf)(g)
    =
    \frac{1}{|\Sigma_s|}
    \sum_{\gamma\in\Sigma_s}
    f(g\gamma).
\end{equation}
To estimate its norm, define
\begin{equation}
    S
    =
    \left\{
        \tau c
        \,\middle|\,
        c\in C_3\setminus\{I_2\}
    \right\},
    \qquad |S|=2.
\end{equation}
For a finite set $E\subset\mathrm{PU}(2)$, write
$A_E=|E|^{-1}\sum_{h\in E}R_h$, where
$(R_hf)(g)=f(gh)$. Expanding the reduced words gives
\begin{equation}
    A_s=A_{C_3}A_S^{s-1}R_\tau A_{C_3}.
\end{equation}
The averaging operators $A_{C_3}$ are contractions and
$R_\tau$ is unitary, including on the mean-zero subspace
$L^2_0$. Equation~(5.3) of
\cite{PARZANCHEVSKI2018869}, with $|S|=2$, bounds
$\|A_S^{s-1}\|_{L^2_0\to L^2_0}$ by
$O(s2^{-(s-1)/2})$. The cited paper uses $f(gh^{-1})$;
our operator is its adjoint and has the same norm. Hence
\begin{equation}
    \|A_s\|_{L^2_0\to L^2_0}=O(s2^{-s/2}).
\end{equation}

Let
\begin{equation}
    B_\varepsilon
    =
    \left\{
        V\in\mathrm{PU}(2)
        \,\middle|\,
        D_\diamond(V,I_2)<\varepsilon
    \right\}.
\end{equation}
Since $\mathrm{PU}(2)$ is a three-dimensional compact Lie group,
\begin{equation}
    \mu(B_\varepsilon)=\Theta(\varepsilon^3)
\end{equation}
as $\varepsilon\rightarrow0$.

The set $\Sigma_s$ is closed under inversion: reverse a
reduced word, use $\tau^{-1}=\tau$ in $\mathrm{PU}(2)$,
and invert its $C_3$ factors. Thus the convention for
$A_s$ agrees with the averaging operator for $\Sigma_s$
in Proposition~3.1 of \cite{PARZANCHEVSKI2018869}.
Applying that proposition to $\Sigma_s$ and $B_\varepsilon$
gives
\begin{align}
    \mu\left(
        \mathrm{PU}(2)
        \setminus
        \bigcup_{\gamma\in\Sigma_s}
        B_\varepsilon\gamma
    \right)
    &\leq
    \frac{
        \left\|A_s\right\|_{L^2_0\rightarrow L^2_0}^2
    }{
        \mu(B_\varepsilon)
    }                                                     \\
    &=
    O\left(
        \frac{s^2}{2^s\varepsilon^3}
    \right).
    \label{eq:even_sphere_covering}
\end{align}

Suppose now that $s$ is even. Every element of $\Sigma_s$ has
an explicit ancilla-free Clifford+T representation containing
exactly $s$ ordinary $T$ gates. This representation is itself
admissible in the definition of $\mathcal{T}_\varepsilon$;
no claim about the minimum unconstrained T-count of the word
is needed. Consequently,
\begin{equation}
    \left\{
        U:\mathcal{T}_\varepsilon(U)>s
    \right\}
    \subseteq
    \mathrm{PU}(2)
    \setminus
    \bigcup_{\gamma\in\Sigma_s}
    B_\varepsilon\gamma.
\end{equation}
It follows from Equation~\eqref{eq:even_sphere_covering} that,
for every even $s$,
\begin{equation}
    \Pr\left[
        \mathcal{T}_\varepsilon(U)>s
    \right]
    =
    O\left(
        \frac{s^2}{2^s\varepsilon^3}
    \right).
    \label{eq:even_t_upper_tail}
\end{equation}

Fix any $\delta>0$ and choose the even integer
\begin{equation}
    s_\varepsilon
    =
    2\left\lceil
        \frac{(3+\delta)\log_2(1/\varepsilon)}{2}
    \right\rceil.
\end{equation}
Then
\begin{equation}
    s_\varepsilon
    \leq
    (3+\delta)\log_2(1/\varepsilon)+2
\end{equation}
and
\begin{equation}
    2^{s_\varepsilon}
    \geq
    \varepsilon^{-(3+\delta)}.
\end{equation}
Substituting this into
Equation~\eqref{eq:even_t_upper_tail} yields
\begin{equation}
    \Pr\left[
        \mathcal{T}_\varepsilon(U)
        >
        (3+\delta)\log_2(1/\varepsilon)+2
    \right]
    =
    O\left(
        \varepsilon^\delta
        \log^2(1/\varepsilon)
    \right).
    \label{eq:even_t_asymptotic_upper}
\end{equation}

For the lower tail, let
$\mathcal{T}^{\mathrm{all}}_\varepsilon(U)$ be the minimum
T-count of an ancilla-free Clifford+T circuit $V$ satisfying
$D_\diamond(U,V)\leq\varepsilon$, with no parity restriction. Every
projective single-qubit Clifford+T circuit of T-count at
most $q$ has a representation of the form
Equation~\eqref{eq:manormal} with at most $q$ syllables.
There are at most
\begin{equation*}
    3\cdot24\sum_{j=0}^{q}2^j=O(2^q)
\end{equation*}
such expressions: the final projective Clifford has $24$
possibilities. Uniqueness of this representation is not
needed for this upper bound on their number. Each expression
approximates a set of targets of Haar measure
$O(\varepsilon^3)$. A union bound, together with
$\mathcal{T}_\varepsilon(U)\geq
\mathcal{T}^{\mathrm{all}}_\varepsilon(U)$, therefore gives
the uniform estimate
\begin{equation} \label{eq:even_t_lower_tail}
    \Pr[\mathcal{T}_\varepsilon(U)\leq r]
    \leq C2^{\lfloor r\rfloor}\varepsilon^3
    \leq C2^r\varepsilon^3
    \qquad(r\geq0).
\end{equation}
This is the usual single-qubit volume-counting argument;
see also \cite{Morisaki:2025mwy}. In particular, for
$0<\delta<3$,
\begin{equation} \label{eq:even_t_asymptotic_lower}
    \Pr\!\left[
        \mathcal{T}_\varepsilon(U)
        \leq(3-\delta)\log_2(1/\varepsilon)
    \right]
    \leq C\varepsilon^\delta.
\end{equation}
Combining this with
Equation~\eqref{eq:even_t_asymptotic_upper}, and absorbing
the logarithmic factor and the additive rounding constant
into a smaller positive slack, proves
Equation~\eqref{eq:single_synthesis_morisaki} in the precise
sense defined in Section~\ref{ssec:asymp}.
\end{proof}

\section{Correspondence between SU(4) and SO(6)} \label{sec:su4so6}
There exists an isomorphism between $\mathrm{SU}(4)/\mathbb Z_2$
and $\mathrm{SO}(6)$. Under this correspondence, $e^{i\theta P}$
rotates the plane associated with a Pauli $P$ through angle
$2\theta$, up to a fixed orientation sign for that plane.
The associated coordinate planes are:
\begin{eqnarray}
    \notag
    I_2\otimes X\leftrightarrow (2,3),~
    I_2\otimes Y\leftrightarrow (1,3),~
    I_2\otimes Z\leftrightarrow (1,2), \\
    \notag
    X\otimes I_2\leftrightarrow (5,6),~
    Y\otimes I_2\leftrightarrow (4,6),~
    Z\otimes I_2\leftrightarrow (4,5), \\
    X\otimes X\leftrightarrow (1,4),~
    Y\otimes X\leftrightarrow (1,5),~
    Z\otimes X\leftrightarrow (1,6), \\
    \notag
    X\otimes Y\leftrightarrow (2,4),~
    Y\otimes Y\leftrightarrow (2,5),~
    Z\otimes Y\leftrightarrow (2,6), \\
    \notag
    X\otimes Z\leftrightarrow (3,4),~
    Y\otimes Z\leftrightarrow (3,5),~
    Z\otimes Z\leftrightarrow (3,6).
\end{eqnarray}

\section{Error control and target distributions for SU(4) synthesis}
\label{sec:SU(4)synthesis}

The error analysis requires the relative signs of the two blocks
to be fixed.  A channel distance between single-qubit unitaries
does not determine this relative sign: $I_2$ and $-I_2$ define
the same single-qubit channel, whereas $I_2\oplus(-I_2)$ is
a nontrivial controlled phase.

\begin{lmm}[Phase-aligned SU(4) approximation]
\label{lmm:su4kinji}
    Use the exact decompositions and notation in
    Equations~\eqref{eq:su4_first_decomposition}--\eqref{eq:su4_final_blocks}.
    Suppose $V_1,\ldots,V_7\in\mathrm{SU}(2)$ satisfy
    \begin{equation*}
        \begin{gathered}
            D_\diamond(V_1V_2^\dagger,G_1)\leq\delta,\qquad
            D_\diamond(V_4V_5^\dagger,G_2)\leq\delta,\\
            D_\diamond(V_3,C)\leq\delta,\qquad
            D_\diamond(V_6\oplus V_7,A\oplus B)\leq\delta,
        \end{gathered}
    \end{equation*}
    where $0<\delta<1$, together with
    \begin{equation}\label{eq:su4_sign_conditions}
        \begin{aligned}
            \operatorname{Re}\operatorname{tr}
                (G_1^\dagger V_1V_2^\dagger)&\geq0,\\
            \operatorname{Re}\operatorname{tr}
                (G_2^\dagger V_4V_5^\dagger)&\geq0.
        \end{aligned}
    \end{equation}
    Then the circuit $U'$ in Equation~\eqref{eq:su4_output_circuit}
    satisfies $D_\diamond(U,U')\leq4\delta$.
\end{lmm}
\begin{proof}
    First consider $R\in\mathrm{SU}(2)$ with
    $\operatorname{Re}\operatorname{tr}R\geq0$.
    Its eigenvalues are $e^{\pm i\alpha}$ with
    $0\leq\alpha\leq\pi/2$.  The convex hull of the eigenvalues
    of $I_2\oplus R$ has minimum modulus $\cos\alpha$, just as
    the convex hull of the eigenvalues of $R$.  The unitary-channel
    diamond-distance formula therefore gives
    \begin{equation}\label{eq:su4_aligned_block_distance}
        D_\diamond(I_4,I_2\oplus R)
        =\sin\alpha=D_\diamond(I_2,R).
    \end{equation}
    The nonnegative-trace condition is essential to this identity.

    For $a,b,v,w\in\mathrm{SU}(2)$, put
    $g=ab^\dagger$ and suppose
    $\operatorname{Re}\operatorname{tr}(g^\dagger vw^\dagger)\geq0$.
    The relative unitary between
    $(v\oplus w)(I_2\otimes v^\dagger a)$ and $a\oplus b$
    has one identity block; its other block has the same trace
    and eigenvalues as $g^\dagger vw^\dagger$ or its adjoint.
    Thus Equation~\eqref{eq:su4_aligned_block_distance} implies
    \begin{equation}\label{eq:su4_pair_replacement}
        \begin{aligned}
            &D_\diamond\!\left(
                (v\oplus w)(I_2\otimes v^\dagger a),
                a\oplus b\right)\\
            &\hspace{25mm}=D_\diamond(vw^\dagger,ab^\dagger).
        \end{aligned}
    \end{equation}
    Applying this identity to the two sign-aligned pairs yields
    \begin{equation*}
        D_\diamond(Q_1K_1,P_1)\leq\delta,\qquad
        D_\diamond(Q_2K_2,A_2\oplus B_2)\leq\delta.
    \end{equation*}

    The Clifford identity $(SH)Z(HS^\dagger)=Y$ gives
    \begin{equation*}
        P_2=(SH\otimes I_2)(A_2\oplus B_2)(HS^\dagger\otimes I_2).
    \end{equation*}
    Moreover, $K_1$ and $P_1$ commute with
    $e^{i\theta_6Z}\otimes I_2$, so the two exact decompositions
    imply
    \begin{equation*}
        U=(C\otimes I_2)P_1K_1^\dagger P_2U_2.
    \end{equation*}
    Starting from Equation~\eqref{eq:su4_output_circuit}, first
    replace $Q_3$ by $K_2(H\otimes I_2)U_2(H\otimes I_2)$.
    Next replace $Q_2K_2$ by $A_2\oplus B_2$, and then replace
    $Q_1$ by $P_1K_1^\dagger$.
    These three replacements cost at most $\delta$ each in
    diamond distance by unitary invariance and the preceding identities.
    The resulting circuit is
    $(V_3\otimes I_2)P_1K_1^\dagger P_2U_2$.
    Replacing $V_3$ by $C$ costs at most another $\delta$ and
    gives $U$.  The triangle inequality proves the claim.
\end{proof}

We next justify the probabilistic estimates used by the construction.
Write a single-qubit special unitary in quaternion coordinates as
\begin{equation*}
    q_0I_2+i(q_1X+q_2Y+q_3Z),\qquad
    q_0^2+q_1^2+q_2^2+q_3^2=1.
\end{equation*}
Normalized Haar measure $\mu_2$ on $\mathrm{SU}(2)$ is uniform
surface measure on this three-sphere.

\begin{lmm}[Distributions of the SU(4) synthesis targets]
\label{lmm:su4_target_distributions}
    There are measurable choices of the decompositions in
    Equations~\eqref{eq:su4_first_decomposition} and
    \eqref{eq:su4_second_decomposition} with the following property.
    For Haar-random $U\in\mathrm{SU}(4)$ and every precision
    $\delta>0$, the targets $G_1,G_2,C,A,B$ in the construction
    have densities with respect to $\mu_2$ bounded uniformly
    in $\delta$.  More precisely, the densities of $G_1$ and
    $G_2$ in quaternion coordinates are
    \begin{equation}\label{eq:su4_target_densities}
        f_1(q)=8q_0^2\mathbf{1}_{\{q_0\geq0\}},\qquad
        f_2(q)=\frac{3\pi}{2}q_0\mathbf{1}_{\{q_0\geq0\}}.
    \end{equation}
    The channel represented by $C$ is Haar-random modulo
    $\{\pm I_2\}$, and the chosen lift has density at most $2$.
    The final blocks $A,B$ are independent Haar-random
    elements of $\mathrm{SU}(2)$, even conditionally on the
    quotient coordinates and synthesis choices made before
    exposing these final blocks.
\end{lmm}
\begin{proof}
    Under the double covering
    $\mathrm{SU}(4)\cong\mathrm{Spin}(6)\to\mathrm{SO}(6)$,
    the subgroups used above form the chain
    \begin{equation*}
        \mathrm{Spin}(6)\supset K_5\cong\mathrm{Spin}(5)
        \supset K_4\cong\mathrm{Spin}(4).
    \end{equation*}
    Their successive quotients are $S^5$ and $S^4$.
    For a Haar-random element of a compact group and a fixed
    measurable section of a homogeneous-space quotient, the quotient
    coordinate is distributed according to invariant probability
    measure and the residual subgroup element is independent Haar.
    To see this, sample the quotient coordinate according to its
    invariant measure and an independent Haar subgroup element;
    their product has a left-invariant probability law, hence is Haar.

    Apply this observation to $U$ and $K_5$.  The quotient
    coordinate $x$ is uniform on $S^5$.  Choose the angles
    $\theta_3,\theta_4,\theta_5\in[-\pi/4,\pi/4]$ so that
    the quaternion coordinates of $G_1=A_1B_1^\dagger$ are
    \begin{equation}\label{eq:su4_palindrome_coordinates}
        \begin{aligned}
            q_0&=\cos(2\theta_3)\cos(2\theta_4)\cos(2\theta_5),\\
            q_1&=\sin(2\theta_3)\cos(2\theta_4)\cos(2\theta_5),\\
            q_2&=\sin(2\theta_4)\cos(2\theta_5),\qquad
            q_3=\sin(2\theta_5).
        \end{aligned}
    \end{equation}
    For completeness, take $R_{ab}(\phi)$ to send
    $e_b$ to $\sin\phi\,e_a+\cos\phi\,e_b$.
    Choose the orientations of the correspondence in
    Appendix~\ref{sec:su4so6} so that $Z\otimes X$,
    $Z\otimes Y$, $Z\otimes Z$, $Y\otimes I_2$, and
    $X\otimes I_2$ generate, respectively,
    $R_{16}(2\theta)$, $R_{26}(2\theta)$,
    $R_{36}(2\theta)$, $R_{46}(2\theta)$, and
    $R_{56}(2\theta)$.
    The sixth column of their product is then exactly
    \begin{equation}\label{eq:su4_first_quotient_column}
        x=\begin{pmatrix}
            q_1\\ q_2\\ q_3\\
            q_0\sin(2\theta_2)\\
            q_0\cos(2\theta_2)\sin(2\theta_1)\\
            q_0\cos(2\theta_2)\cos(2\theta_1)
        \end{pmatrix}.
    \end{equation}
    Thus $q_j=x_j$ for $1\leq j\leq3$ and
    $q_0=(x_4^2+x_5^2+x_6^2)^{1/2}$.
    The remaining angles $\theta_1\in[0,\pi)$ and
    $\theta_2\in[-\pi/4,\pi/4]$ describe the direction
    of $(x_4,x_5,x_6)$, which is uniform on $S^2$ and independent
    of $G_1$.  Choices on the coordinate singularities may be made
    arbitrarily, since those sets have measure zero.

    For uniform $x\in S^5$, the density of
    $(x_1,x_2,x_3)$ in the unit three-ball is proportional to
    $(1-x_1^2-x_2^2-x_3^2)^{1/2}$.  By contrast, its density
    under uniform surface measure on either hemisphere of $S^3$
    is proportional to
    $(1-x_1^2-x_2^2-x_3^2)^{-1/2}$.
    Their ratio is proportional to $q_0^2$.
    Since $\int_{q_0\geq0}q_0^2\,d\mu_2=1/8$,
    the first formula in Equation~\eqref{eq:su4_target_densities}
    follows.

    Conditional on the first quotient coordinate, $U_1$ is
    Haar in $K_5$.  The adaptive correction
    $K_1=I_2\otimes(V_1^\dagger A_1)$ belongs to $K_5$
    and is determined by that coordinate and the fixed synthesis rule:
    the subroutine producing $V_1\oplus V_2$ uses only
    $G_1$ and $\delta$.
    Left invariance therefore implies that $K_1U_1$ remains
    conditionally Haar in $K_5$, with a law independent of
    both the first quotient coordinate and $\delta$.
    Its quotient by $K_4$ is uniform on $S^4$.
    Choose $\theta_7,\theta_8,\theta_9$ in the same principal
    ranges as above, and denote the quaternion coordinates of $G_2$
    by $(r_0,r_1,r_2,r_3)$.
    In the chosen orientation, the Lie brackets give
    $Z\otimes I_2\leftrightarrow R_{45}(-2\theta)$ and
    $Y\otimes X,Y\otimes Y,Y\otimes Z\leftrightarrow
    R_{15}(2\theta),R_{25}(2\theta),R_{35}(2\theta)$.
    The quotient column within the first five coordinates is
    therefore
    \begin{equation}\label{eq:su4_second_quotient_column}
        y=(r_1,r_2,r_3,-r_0\sin(2\theta_6),
        r_0\cos(2\theta_6))^{\mathsf T}.
    \end{equation}
    Thus the scalar coordinate $r_0$ is the norm of the last
    two coordinates of a uniform point on $S^4$.
    Projection onto the first three coordinates has
    constant density on the unit three-ball.  Comparing again
    with the hemisphere measure on $S^3$ gives density proportional
    to $q_0$.  The identity
    $\int_{q_0\geq0}q_0\,d\mu_2=2/(3\pi)$ gives the second
    formula in Equation~\eqref{eq:su4_target_densities}.
    The remaining angle $\theta_6\in[0,\pi)$ is uniform,
    independently of $G_2$ and of the first quotient coordinate.

    The independent uniform directions on $S^2$ and $S^1$
    give $(\theta_1,\theta_2,\theta_6)$ a joint density
    proportional to $\cos(2\theta_2)$ on
    \begin{equation*}
        [0,\pi)\times[-\pi/4,\pi/4]\times[0,\pi).
    \end{equation*}
    This is the Euler-angle density of Haar measure on
    $\mathrm{SU}(2)/\{\pm I_2\}$ for
    $C=e^{i\theta_1X}e^{i\theta_2Y}e^{i\theta_6Z}$.
    A measurable choice of one of the two lifts has Haar density
    at most $2$ on $\mathrm{SU}(2)$.

    Finally, $U_2$ is conditionally Haar in $K_4$ and independent
    of the preceding quotient coordinates.  Conjugation by
    $H\otimes I_2$ identifies $K_4$ with
    $\mathrm{SU}(2)\oplus\mathrm{SU}(2)$.
    The subroutine producing $V_4\oplus V_5$ uses only
    $G_2$ and $\delta$.
    Consequently the adaptive correction
    $K_2=I_2\otimes(V_4^\dagger A_2)$
    belongs to this last subgroup and is fixed after conditioning
    on the preceding data.  It preserves conditional Haar measure.
    Consequently the $A,B$ in Equation~\eqref{eq:su4_final_blocks}
    have the product Haar law, as asserted.
\end{proof}

\section{Proofs of the auxiliary lemmas}
\label{app:clean_auxiliary_proofs}

This appendix proves the lemmas of
Section~\ref{sssec:clean_ancilla_auxiliary} in their order of
appearance, from Lemma~\ref{lmm:clean_arithmetic} to
Lemma~\ref{lmm:clean_small_ball}.
We use the notation introduced in that section.

Recall that
$k=\mathbb{Q}(\sqrt{2})$, $E=k(i)=\mathbb{Q}(\zeta)$,
and $\zeta=e^{i\pi/4}$.
Complex conjugation is denoted by a bar, whereas
$\sigma$ fixes $i$ and sends $\sqrt{2}$ to $-\sqrt{2}$;
these two automorphisms commute.
For $z\in E$, let $N_E(z)=|N_{E/\mathbb{Q}}(z)|$
be the absolute field norm.
For a nonzero integral ideal $\mathfrak{a}$ of
$\mathcal{O}_E$, write
$N_E(\mathfrak{a})=|\mathcal{O}_E/\mathfrak{a}|$.
Thus ideal norms measure the indices of the corresponding
additive subgroups. In particular, for an ideal
$\mathfrak{a}\subseteq\mathcal{O}_k$,
\begin{equation}
    N_E(\mathfrak{a}\mathcal{O}_E)=N_k(\mathfrak{a})^2,
\end{equation}
since $\mathcal{O}_E=\mathcal{O}_k\oplus\zeta\mathcal{O}_k$.

\subsection{Proof of Lemma~\ref{lmm:clean_arithmetic}: arithmetic constraints}
\label{app:proof_clean_arithmetic}

\begin{proof}[Proof of Lemma~\ref{lmm:clean_arithmetic}]
    Let $W$ be the full circuit unitary, acting on $n+1+a$
    qubits, and define
    \begin{equation}
        J_a|\psi\rangle=|\psi\rangle\otimes|0^a\rangle.
    \end{equation}
    The clean-ancilla condition is
    \begin{equation}
        WJ_a=J_aV.
    \end{equation}
    Since $V$ is unitary, it also gives
    $W^\dagger J_a=J_aV^\dagger$.
    Consequently,
    \begin{equation}
        VAV^\dagger
        =
        J_a^\dagger W(A\otimes I_{2^a})W^\dagger J_a
    \end{equation}
    for every operator $A$ on the logical register.

    Let $\mathcal{P}_m=\{I_2,X,Y,Z\}^{\otimes m}$ be the
    Hermitian Pauli basis on $m$ qubits.
    For $P,Q\in\mathcal{P}_{n+1}$, the preceding identity gives
    \begin{equation}
    \label{eq:clean_ptm_compression}
        \widehat{V}_{P,Q}
        =
        \sum_{A\in\{I_2,Z\}^{\otimes a}}
        \widehat{W}_{P\otimes A,\,
                     Q\otimes I_2^{\otimes a}}.
    \end{equation}
    Indeed,
    $\langle0^a|A|0^a\rangle$ equals $1$ for
    $A\in\{I_2,Z\}^{\otimes a}$ and vanishes for all other
    Hermitian Pauli operators.

    Clifford conjugation acts as a signed permutation of the
    Pauli basis, and the channel representation of each $T$
    gate has entries in $(\sqrt{2})^{-1}\mathcal{O}_k$.
    Hence
    \begin{equation}
        (\sqrt{2})^t\widehat{W}
        \in M_{4^{n+1+a}}(\mathcal{O}_k).
    \end{equation}
    Equation~\eqref{eq:clean_ptm_compression} introduces no
    additional denominators, proving
    Equation~\eqref{eq:clean_ptm_denominator}.

    The entries of $W$, and therefore those of
    $V=J_a^\dagger WJ_a$, belong to $\mathcal{R}$.
    To obtain the sharper denominator bound, define
    \begin{equation}
        v=\sum_{j=1}^{2^{n+1}}(Ve_j)\otimes e_j.
    \end{equation}
    The channel representation satisfies
    \begin{equation}
        vv^\dagger
        =
        \frac{1}{2^{n+1}}
        \sum_{P,Q\in\mathcal{P}_{n+1}}
        \widehat{V}_{P,Q}\,P\otimes Q^{\mathsf T}.
    \end{equation}
    Thus
    \begin{equation}
        (\sqrt{2})^{t+2n+2}vv^\dagger
        \in M_{2^{2n+2}}(\mathcal{O}_E).
    \end{equation}

    Set $\varrho=1-\zeta$.
    The ideal $(\varrho)$ is the unique prime ideal of
    $\mathcal{O}_E$ above $2$, and it induces the valuation function $v_\varrho$ on $E$.
    Since $\varrho^2=\sqrt2(1-\sqrt2)\zeta$, $v_\varrho(\sqrt{2})=2$.
    Complex conjugation preserves this valuation.
    Applying the preceding integrality statement to a
    diagonal entry of $vv^\dagger$ gives
    \begin{equation}
        2v_\varrho(V_{jk})
        \geq -2(t+2n+2).
    \end{equation}
    Since
    $2s=2\lceil t/2\rceil+2n+2\geq t+2n+2$,
    every entry of $\gamma V$ is integral at $(\varrho)$.
    It is integral at all other finite primes because
    $V$ has entries in $\mathcal{R}$.
    This proves Equation~\eqref{eq:clean_matrix_denominator}.

    Finally, $\sigma$ commutes with complex conjugation.
    Applying it to $V^\dagger V=I_{2^{n+1}}$ shows that
    $\sigma(V)$ is unitary.

    For the determinant assertion, let $Y$ be any unitary
    matrix with entries in $\mathcal{R}$ and set
    $\xi=\det(Y)$.
    Then $\xi\overline{\xi}=1$.
    Since conjugation preserves $v_\varrho$, this identity
    implies $v_\varrho(\xi)=0$.
    At every other finite prime, $\xi$ is already integral.
    Hence $\xi\in\mathcal{O}_E$.
    All its complex conjugates have absolute value $1$,
    so $\xi$ is a root of unity.
    The roots of unity in $E$ are precisely
    $\{\zeta^j:0\leq j<8\}$.
    Therefore $\det(Y)^8=1$.
\end{proof}

\subsection{Proof of Lemma~\ref{lmm:clean_block_geometry}: block constraints}
\label{app:proof_clean_block_geometry}

\begin{proof}[Proof of Lemma~\ref{lmm:clean_block_geometry}]
    Since $D_\diamond(U,V)<1$, the eigenvalues of
    $U^\dagger V$ lie in an arc of angular width less than
    $\pi$.
    Choose their arguments
    $\theta_1,\ldots,\theta_{2^{n+1}}$ in such an arc, and let
    \begin{equation}
        \overline{\theta}
        =\frac{1}{2^{n+1}}\sum_{\ell=1}^{2^{n+1}}\theta_\ell,
        \qquad
        \eta=e^{i\overline{\theta}}.
    \end{equation}
    The chord joining the endpoints of the smallest such
    arc has length at most $2\varepsilon$.
    Since $\overline{\theta}$ lies between the smallest and
    largest arguments,
    \begin{equation}
        \|V-\eta U\|_{\mathrm{op}}
        =
        \max_\ell
        |e^{i\theta_\ell}-e^{i\overline{\theta}}|
        \leq2\varepsilon.
    \end{equation}
    Also,
    \begin{equation}
        \eta^{2^{n+4}}
        =
        \det(U^\dagger V)^8
        =
        \det(V)^8
        =1.
    \end{equation}

    For every unit vector $|\psi\rangle$,
    \begin{equation}
        1-|\langle U\psi,V\psi\rangle|^2
        \leq D_\diamond(U,V)^2
        \leq\varepsilon^2.
    \end{equation}
    Taking $|\psi\rangle=e_i\otimes|\xi\rangle$ with
    $\|\xi\|=1$ gives
    \begin{equation}
        \sum_{j\neq i}\|\blk{V}{j}{i}\xi\|^2
        =
        1-\|\blk{V}{i}{i}\xi\|^2
        \leq\varepsilon^2.
    \end{equation}
    Taking the supremum over $\xi$ proves
    Equations~\eqref{eq:clean_block_leakage}
    and~\eqref{eq:clean_diagonal_singular_values}.

    Finally, set
    \begin{equation}
        R=\eta^{-1}U^\dagger V=I_{2^{n+1}}+E.
    \end{equation}
    Then $\|E\|_{\mathrm{op}}\leq2\varepsilon$.
    Unitarity gives $E+E^\dagger=-EE^\dagger$, so
    \begin{equation}
        \|\blk{R}{i}{j}+\blk{R}{j}{i}^\dagger\|_{\mathrm{op}}
        \leq\|EE^\dagger\|_{\mathrm{op}}
        \leq4\varepsilon^2
        \qquad (i\neq j).
    \end{equation}
    Multiplication by $\eta$ yields
    Equation~\eqref{eq:clean_block_cancellation}.
\end{proof}

\subsection{Proof of Lemma~\ref{lmm:clean_norm_equations}: norm equations}
\label{app:proof_clean_norm_equations}

\begin{proof}[Proof of Lemma~\ref{lmm:clean_norm_equations}]
    If $z\overline{z}=\alpha$, then the principal ideal
    $(z)$ divides $(\alpha)$ in $\mathcal{O}_E$.
    Two generators $z_1,z_2$ of the same ideal satisfying
    this equation differ by a unit $u$ with
    $u\overline{u}=1$.
    All conjugates of $u$ have absolute value $1$, so there
    are at most eight such units.

    The number of ideal divisors of $(\alpha)$ is bounded
    by a fixed power of the ordinary divisor function
    evaluated at
    $N_E(\alpha)=N_k(\alpha)^2$ \cite{KhanLeiMilicevic2019}.
    The usual divisor bound therefore proves
    Equation~\eqref{eq:clean_norm_one}.

    For the second assertion, every possible first coordinate
    satisfies
    \begin{equation}
        |u_1|\leq\sqrt{\beta},
        \qquad
        |\sigma(u_1)|\leq\sqrt{\sigma(\beta)}.
    \end{equation}
    Let $\rho=1+\sqrt{2}$, so that
    $|\sigma(\rho)|=\rho^{-1}$.
    Choose an integer $m$, depending only on $\beta$, such
    that multiplication by $\rho^m$ makes both bounds
    $O(N_k(\beta)^{1/4})$.
    This multiplication is a bijection of $\mathcal{O}_E$.
    Expressing $\rho^m u_1$ in the integral basis
    $1,\zeta,\zeta^2,\zeta^3$ shows that there are
    $O(N_k(\beta))$ possible values of $u_1$.

    For each such $u_1$, put
    $\alpha=\beta-u_1\overline{u_1}$.
    If a nonzero second coordinate exists, then $\alpha$
    is totally positive and
    $N_k(\alpha)\leq N_k(\beta)$.
    Equation~\eqref{eq:clean_norm_one} bounds the number
    of second coordinates by
    $C_\nu N_k(\beta)^\nu$.
    If $\alpha=0$, the only second coordinate is $u_2=0$.
    Multiplying by the number of possible first coordinates
    proves Equation~\eqref{eq:clean_norm_two}.
\end{proof}

\subsection{Proof of Lemma~\ref{lmm:clean_fixed_gram}: fixed Gram matrix}
\label{app:fixed_gram}

We use the norm-equation estimates of
Lemma~\ref{lmm:clean_norm_equations} to count matrices
with a fixed Gram matrix.

\begin{proof}[Proof of Lemma~\ref{lmm:clean_fixed_gram}]
Put
\begin{equation}
    \begin{gathered}
        r(G)=\#\{M\in M_2(\mathcal{O}_E):M^\dagger M=G\},\\
        C=N_k(\mathfrak{c}(G)),
        \qquad D=N_k(\det G).
    \end{gathered}
\end{equation}
We will prove the stronger estimate
\begin{equation}
    r(G)\ll_\nu C D^\nu
    \qquad\text{for every }\nu>0.
    \label{eq:fixed_gram_power_bound}
\end{equation}
All implied constants below depend only on the fixed
fields $k,E$ and on the displayed subscripts, never on $G$.
If $r(G)=0$, there is nothing to prove.

\paragraph{The ring of integers and the content.}
We first record why $\mathcal{O}_E$ is a principal ideal
domain, as needed for the Smith normal form below.
The integral basis $1,\zeta,\zeta^2,\zeta^3$ and the
polynomial $X^4+1$ give $|\operatorname{disc}(E)|=256$
\cite[Proposition~6.2]{MilneANT2020}.
Minkowski's ideal-class bound
\cite[Theorem~4.3]{MilneANT2020} implies that every ideal
class has an integral representative of norm at most
\begin{equation}
    \left(\frac{4}{\pi}\right)^2
    \frac{4!}{4^4}\sqrt{256}
    =\frac{24}{\pi^2}<3.
\end{equation}
An ideal of norm $1$ is the unit ideal. The ideal
$(1-\zeta)$ has norm $\Phi_8(1)=2$, where
$\Phi_8(X)=X^4+1$.
It is the only ideal of norm $2$: any quotient of order
$2$ is $\mathbb{F}_2$, and the image of the unit $\zeta$
in that quotient must be $1$.
Thus such an ideal contains $(1-\zeta)$ and equals it,
as both ideals have norm $2$.
Both possibilities are principal, proving the assertion.

Write
\begin{equation}
    G=\begin{pmatrix}a&b\\\overline{b}&c\end{pmatrix},
    \qquad b=x+y\zeta,
    \qquad a,c,x,y\in\mathcal{O}_k.
\end{equation}
Because $\overline{\zeta}=\sqrt{2}-\zeta$, the definition
of content in the main text gives
\begin{equation}
    \mathfrak{c}(G)=(a,c,x,y),
    \qquad
    \det G=ac-x^2-\sqrt{2}xy-y^2\in\mathfrak{c}(G)^2.
\end{equation}
Taking ideal norms in the inclusion
$(\det G)\subseteq\mathfrak{c}(G)^2$ therefore yields
\begin{equation}
    C^2\leq D.
    \label{eq:fixed_gram_content_det}
\end{equation}

\paragraph{Fixing the determinant.}
For each possible determinant $\delta\in\mathcal{O}_E$,
let
\begin{equation}
    \mathcal{R}_\delta(G)
    =\{M\in M_2(\mathcal{O}_E):M^\dagger M=G,
      \ \det M=\delta\}.
\end{equation}
If this set is nonempty, fix
$M_0=(v\ w)\in\mathcal{R}_\delta(G)$.
For $M\in\mathcal{R}_\delta(G)$, the matrix
$U=MM_0^{-1}$ has entries in $E$ and satisfies
\begin{equation}
    U^\dagger U
    =(M_0^{-1})^\dagger G M_0^{-1}=I_2,
    \qquad \det U=1.
\end{equation}
Thus $U\in\mathrm{SU}(2)$.
For a column vector $q=(q_1,q_2)^{\mathsf T}\in E^2$,
define its conjugate orthogonal companion by
\begin{equation}
    q^\sharp=
    \begin{pmatrix}-\overline{q_2}\\\overline{q_1}\end{pmatrix}.
\end{equation}
It obeys $q^\dagger q^\sharp=0$ and
$\det(q\ q^\sharp)=q^\dagger q$.
Moreover,
\begin{equation}
    (Uq)^\sharp=Uq^\sharp
    \qquad(U\in\mathrm{SU}(2)),
    \label{eq:fixed_gram_sharp}
\end{equation}
as follows directly by writing
$U=\left(\begin{smallmatrix}
\alpha&-\overline{\gamma}\\
\gamma&\overline{\alpha}
\end{smallmatrix}\right)$.

Consider the $\mathcal{O}_E$-module
\begin{equation}
    \mathcal{L}
    =\mathcal{O}_E v+\mathcal{O}_E v^\sharp
     +\mathcal{O}_E w+\mathcal{O}_E w^\sharp
    \subseteq\mathcal{O}_E^2.
    \label{eq:fixed_gram_lattice}
\end{equation}
The vectors $v,v^\sharp$ are linearly independent over
$E$, since $\det(v\ v^\sharp)=a>0$.
Hence $\mathcal{L}$ has rank $2$ over $\mathcal{O}_E$.
Its purpose is to provide integral vectors whose images
remain integral for every admissible $U$:
the columns $Uv,Uw$ of $M$ are integral, and
Equation~\eqref{eq:fixed_gram_sharp} gives
\begin{equation}
    U\mathcal{L}\subseteq\mathcal{O}_E^2.
    \label{eq:fixed_gram_integrality}
\end{equation}

\paragraph{The index of the lattice.}
Let $\mathfrak{d}$ be the ideal generated by the
$2\times2$ minors of $(v\ v^\sharp\ w\ w^\sharp)$.
These six minors are $a,\delta,\overline{b},-b,
\overline{\delta},c$, so
\begin{equation}
    \mathfrak{d}
    =(a,c,b,\overline{b},\delta,\overline{\delta})
      \mathcal{O}_E.
\end{equation}
Smith normal form over the principal ideal domain
$\mathcal{O}_E$, together with the determinantal-ideal
formula, gives
\cite[Theorems~2.1(4) and~2.4]{Stanley2016Smith}
\begin{equation}
    [\mathcal{O}_E^2:\mathcal{L}]=N_E(\mathfrak{d}).
    \label{eq:fixed_gram_index}
\end{equation}
Indeed, if the two nonzero Smith entries are $d_1,d_2$,
then $\mathfrak{d}=(d_1d_2)$ and the quotient is
$\mathcal{O}_E/(d_1)\oplus\mathcal{O}_E/(d_2)$.

Set $\varpi=\sqrt{2}$.
Since $\zeta-\overline{\zeta}=i\varpi$, we have
\begin{equation}
    \varpi y=-i(b-\overline{b})\in\mathfrak{d},
    \qquad
    \varpi x=\varpi b-\zeta\varpi y\in\mathfrak{d}.
\end{equation}
Together with $a,c\in\mathfrak{d}$, this shows that
$(\varpi)\mathfrak{c}(G)\mathcal{O}_E\subseteq\mathfrak{d}$.
Because $N_k((\varpi))=2$, Equation~\eqref{eq:fixed_gram_index}
implies
\begin{equation}
    [\mathcal{O}_E^2:\mathcal{L}]
    \leq N_E((\varpi)\mathfrak{c}(G)\mathcal{O}_E)
    =4C^2.
    \label{eq:fixed_gram_index_bound}
\end{equation}

\paragraph{A short vector and an injective map.}
Use the Minkowski embedding
\begin{equation}
    \iota:E^2\longrightarrow\mathbb{C}^2\times\mathbb{C}^2
                    \simeq\mathbb{R}^8,
    \qquad q\longmapsto(q,\sigma(q)).
\end{equation}
The image $\iota(\mathcal{O}_E^2)$ is a full lattice in
$\mathbb{R}^8$ by applying
\cite[Proposition~4.26]{MilneANT2020} to each coordinate.
Since $\mathcal{L}$ has finite index in $\mathcal{O}_E^2$,
its image is also a full lattice. Its covolume, meaning the volume of a
fundamental parallelepiped, is
\begin{equation}
    \operatorname{covol}(\iota(\mathcal{L}))
    =[\mathcal{O}_E^2:\mathcal{L}]
      \operatorname{covol}(\iota(\mathcal{O}_E^2))
    \ll C^2.
\end{equation}
Apply Minkowski's convex-body theorem
\cite[Theorem~4.19]{MilneANT2020} to a Euclidean ball in
this eight-dimensional space. It gives a nonzero
$q_0\in\mathcal{L}$ such that
\begin{equation}
    \|\iota(q_0)\|_2
    \ll\operatorname{covol}(\iota(\mathcal{L}))^{1/8}
    \ll C^{1/4}.
\end{equation}
The element $\beta=q_0^\dagger q_0\in\mathcal{O}_k$ is
totally positive: its two real embeddings are the squared
lengths of the nonzero vectors $q_0$ and $\sigma(q_0)$.
Consequently,
\begin{equation}
    N_k(\beta)=\beta\sigma(\beta)
    \leq\frac{(\beta+\sigma(\beta))^2}{4}
    =\frac{\|\iota(q_0)\|_2^4}{4}
    \ll C.
    \label{eq:fixed_gram_short_norm}
\end{equation}

For every $M\in\mathcal{R}_\delta(G)$, define
$z=Uq_0$ with $U=MM_0^{-1}$ as above.
Equation~\eqref{eq:fixed_gram_integrality} and unitarity give
\begin{equation}
    z\in\mathcal{O}_E^2,
    \qquad z^\dagger z=\beta.
\end{equation}
This assignment is injective, since $z$ determines $U$
and hence $M$ through
\begin{equation}
    U=(z\ z^\sharp)(q_0\ q_0^\sharp)^{-1},
    \qquad M=UM_0.
\end{equation}
The inverse exists because
$\det(q_0\ q_0^\sharp)=\beta\neq0$.
The two-variable estimate
\eqref{eq:clean_norm_two} and
Equation~\eqref{eq:fixed_gram_short_norm} therefore imply,
for every $\theta>0$,
\begin{equation}
    \#\mathcal{R}_\delta(G)
    \ll_\theta N_k(\beta)^{1+\theta}
    \ll_\theta C^{1+\theta}.
    \label{eq:fixed_gram_fixed_det}
\end{equation}
This bound is independent of the choices of $\delta$
and $M_0$.

\paragraph{Summing over determinants and uniformity.}
Every possible $\delta$ satisfies
$\delta\overline{\delta}=\det G$.
Since $G$ and $\sigma(G)$ are positive definite,
$\det G$ is totally positive.
The one-variable estimate~\eqref{eq:clean_norm_one}
bounds the number of determinants by $O_\theta(D^\theta)$.
Combining this with
Equations~\eqref{eq:fixed_gram_content_det}
and~\eqref{eq:fixed_gram_fixed_det}, we obtain
\begin{equation}
    r(G)\ll_\theta D^\theta C^{1+\theta}
    \leq C D^{3\theta/2}.
\end{equation}
Taking $\theta=2\nu/3$ proves
Equation~\eqref{eq:fixed_gram_power_bound}.

Finally, fix $A>0$ and let $F_A(s)$ be the supremum of
$r(G)/N_k(\mathfrak{c}(G))$ over the matrices in the
lemma with $N_k(\det G)\leq2^{As}$, enlarged to at least
$1$. Equation~\eqref{eq:fixed_gram_power_bound} gives,
for each $\nu>0$, a constant $K_\nu$ independent of $G,s$
such that
\begin{equation}
    F_A(s)\leq\max\{1,K_\nu 2^{A\nu s}\},
    \qquad
    0\leq\limsup_{s\to\infty}
        \frac{\log_2 F_A(s)}{s}\leq A\nu.
\end{equation}
As $\nu>0$ is arbitrary, $F_A(s)=2^{o(s)}$.
This proves the required bound with a factor uniform over
$G$ and with dependence on $A$ permitted.
\end{proof}

\subsection{Proof of Lemma~\ref{lmm:clean_trace_determinant}: fixed trace and determinant}
\label{app:proof_clean_trace_determinant}

\begin{proof}[Proof of Lemma~\ref{lmm:clean_trace_determinant}]
    First suppose that $z\neq0$.
    For each admissible $M$, put $G=M^\dagger M$.
    Its content $\mathfrak{c}(G)$ contains $(\alpha)$,
    so there are only $2^{o(s)}$ possible content ideals \cite{KhanLeiMilicevic2019}.

    Fix one such ideal $\mathfrak{c}$.
    Since $\mathcal{O}_k$ is a principal ideal domain and
    has units with all possible sign patterns at its two
    real embeddings, choose a totally positive generator
    $h$ of $\mathfrak{c}$.
    Set $\beta=\alpha/h$.
    Every Gram matrix with content exactly $\mathfrak{c}$
    has the form
    \begin{equation}
        \frac{G}{h}
        =
        \begin{pmatrix}
            x&w\\
            \overline{w}&\beta-x
        \end{pmatrix},
        \qquad
        x\in\mathcal{O}_k,\quad w\in\mathcal{O}_E,
    \end{equation}
    where
    \begin{equation}
        w\overline{w}
        =
        x(\beta-x)-\frac{z\overline{z}}{h^2}.
    \end{equation}
    Positive definiteness at both embeddings implies
    \begin{equation}
        0<x<\beta,
        \qquad
        0<\sigma(x)<\sigma(\beta).
    \end{equation}

    For distinct possible values $x,x'$, the nonzero
    algebraic integer $x-x'$ satisfies
    $|N_k(x-x')|\geq1$.
    Hence their spacing at the first embedding is at least
    $1/\sigma(\beta)$.
    The number of possible values of $x$ is therefore
    \begin{equation}
        O(N_k(\beta))
        =
        O\!\left(
            \frac{N_k(\alpha)}{N_k(\mathfrak{c})}
        \right).
    \end{equation}
    For each $x$, Lemma~\ref{lmm:clean_norm_equations}
    gives at most $2^{o(s)}$ possible values of $w$.
    This also includes the case $w=0$.
    Thus the number of possible Gram matrices with content
    exactly $\mathfrak{c}$ is at most
    \begin{equation}
        2^{o(s)}
        \frac{N_k(\alpha)}{N_k(\mathfrak{c})}.
    \end{equation}

    To apply Lemma~\ref{lmm:clean_fixed_gram}, note that
    the arithmetic--geometric mean inequality at both
    embeddings gives
    \begin{equation}
        0<\det G\leq\frac{\alpha^2}{4},
        \qquad
        0<\sigma(\det G)
        \leq\frac{\sigma(\alpha)^2}{4}.
    \end{equation}
    Consequently,
    \begin{equation}
        N_k(\det G)
        \leq\frac{N_k(\alpha)^2}{16}
        \leq2^{2As}.
    \end{equation}
    For each such $G$, that lemma bounds the number of
    matrices $M$ by
    $2^{o(s)}N_k(\mathfrak{c})$.
    The content factors cancel, and summing over the
    $2^{o(s)}$ possible content ideals proves the claim
    when $z\neq0$.

    Now suppose that $z=0$.
    Since $\alpha>0$, $M$ has rank one over $E$.
    Choose a nonzero column $c_j$ of $M$.
    Since $\mathcal{O}_E$ is a principal ideal domain,
    the ideal generated by the two entries of $c_j$
    has a generator $g\neq0$.
    Writing
    \begin{equation}
        u=g^{-1}c_j
        =\begin{pmatrix}u_1\\u_2\end{pmatrix}
        \in\mathcal{O}_E^2,
    \end{equation}
    we have
    $u_1\mathcal{O}_E+u_2\mathcal{O}_E=\mathcal{O}_E$.
    Thus there exist $p,q\in\mathcal{O}_E$ such that
    \begin{equation}
        pu_1+qu_2=1.
    \end{equation}

    Because $M$ has rank one, each column $c_\ell$
    can be written as $c_\ell=\lambda_\ell u$
    for some $\lambda_\ell\in E$.
    The preceding identity gives
    \begin{equation}
        \lambda_\ell
        =
        p(c_\ell)_1+q(c_\ell)_2
        \in\mathcal{O}_E.
    \end{equation}
    Therefore
    \begin{equation}
        M=u\begin{pmatrix}\lambda_1&\lambda_2\end{pmatrix}
        =uv^\dagger,
        \qquad
        v=
        \begin{pmatrix}
            \overline{\lambda_1}\\
            \overline{\lambda_2}
        \end{pmatrix}
        \in\mathcal{O}_E^2.
    \end{equation}
    Both $u$ and $v$ are nonzero because $M\neq0$.
    Put
    \begin{equation}
        a=u^\dagger u,
        \qquad
        b=v^\dagger v.
    \end{equation}
    Then $a,b\in\mathcal{O}_k$ are totally positive and
    satisfy $ab=\alpha$.
    There are $2^{o(s)}$ possible pairs of ideals
    $((a),(b))$ \cite{KhanLeiMilicevic2019}.

    Fix a pair of ideals
    $\mathfrak{a}=(a)$ and $\mathfrak{b}=(b)$
    that occurs among these factorizations.
    Choose one pair of totally positive generators
    $a_0,b_0$ satisfying
    \begin{equation}
        (a_0)=\mathfrak{a},
        \qquad
        (b_0)=\mathfrak{b},
        \qquad
        a_0b_0=\alpha.
    \end{equation}

    Let $\rho=1+\sqrt{2}$.
    The units of $\mathcal{O}_k$ are
    \begin{equation}
        \mathcal{O}_k^\times
        =
        \{\pm\rho^m:m\in\mathbb{Z}\}.
    \end{equation}
    Since $\rho>0$ and $\sigma(\rho)<0$,
    the totally positive units are precisely
    $\rho^{2m}$ with $m\in\mathbb{Z}$.
    Now $(a)=(a_0)$ implies that $a/a_0$ is a unit.
    As $a$ and $a_0$ are both totally positive,
    there exists $m\in\mathbb{Z}$ such that
    \begin{equation}
        a=\rho^{2m}a_0.
    \end{equation}
    The identities $ab=a_0b_0=\alpha$ then imply
    $b=\rho^{-2m}b_0$.

    Replace the factors $u,v$ by
    \begin{equation}
        u'=\rho^{-m}u,
        \qquad
        v'=\rho^m v.
    \end{equation}
    Both vectors still belong to $\mathcal{O}_E^2$,
    because $\rho$ is a unit.
    Moreover, $\rho$ is real, so
    \begin{equation}
        u'v'^\dagger=uv^\dagger=M,
        \qquad
        u'^\dagger u'=a_0,
        \qquad
        v'^\dagger v'=b_0.
    \end{equation}
    Thus, for this ideal pair, it suffices to count vector
    pairs with these two fixed squared norms.

    By Lemma~\ref{lmm:clean_norm_equations}, the number of
    such pairs, and hence the number of matrices admitting
    a factorization with this ideal pair, is at most
    \begin{equation}
        2^{o(s)}N_k(a_0)N_k(b_0)
        =
        2^{o(s)}N_k(\alpha).
    \end{equation}
    Here $N_k(a_0),N_k(b_0)\leq N_k(\alpha)\leq2^{As}$,
    since $a_0b_0=\alpha$ and their norms are positive integers.
    Summing over the $2^{o(s)}$ possible ideal pairs
    proves the rank-one case.
\end{proof}

\subsection{Proof of Lemma~\ref{lmm:clean_small_ball}: small-ball bound}
\label{app:proof_clean_small_ball}

\begin{proof}[Proof of Lemma~\ref{lmm:clean_small_ball}]
    By singular value decomposition and Haar invariance,
    we may replace $P$ by
    $D=\operatorname{diag}(r,u)$.
    The scalar phases of the two unitary factors can be
    absorbed into $F$.

    If the event in question is empty, there is nothing
    to prove.
    Otherwise, choose $L_0,R_0\in\mathrm{SU}(2)$ with
    $\|L_0DR_0-F\|_{\mathrm{op}}\leq e$.
    The triangle inequality and Haar invariance reduce
    the required upper bound to the probability of
    \begin{equation}
        \|LDR-D\|_{\mathrm{op}}\leq2e.
    \end{equation}

    On this event,
    \begin{equation}
        \|R^\dagger D^2R-D^2\|_{\mathrm{op}}
        \le \left\|R^\dagger DL^\dagger(LDR-D)\right\|_{\mathrm{op}} +
        \left\|(R^\dagger DL^\dagger-D)D\right\|_{\mathrm{op}}
        \leq4er.
    \end{equation}
    Writing $z=R_{21}$, we have
    \begin{equation}
        \|R^\dagger D^2R-D^2\|_{\mathrm{op}}
        =
        (r^2-u^2)|z|.
    \end{equation}
    If $b>0$, it follows that $|z|\leq4e/b$.
    Since $|R_{21}|^2$ is uniform on $[0,1]$ under
    Haar measure\cite{Cappellini_2009}, the measure of admissible $R$ is at most
    \begin{equation}
        \min\{1,(4e/b)^2\}.
    \end{equation}
    For $b=0$, we use the trivial bound $1$.

    Fix an $R$ for which the event is nonempty, and put
    $w=DRe_1$.
    If $r\geq4e$, then
    \begin{equation}
        \|w\|\geq r-2e\geq r/2.
    \end{equation}
    For Haar-random $L$, the vector $Lw$ is uniform on
    the three-dimensional sphere of radius $\|w\|$
    in $\mathbb{C}^2$.
    The condition
    $\|Lw-re_1\|\leq2e$ therefore has probability at most
    $C(e/r)^3$.
    Multiplying this conditional bound by the measure
    of admissible $R$ proves the assertion for $r\geq4e$.

    If $r<4e$, then $b\leq r$ and the product of the
    two minimum factors in
    Equation~\eqref{eq:clean_small_ball}
    is at least $4^{-5}$.
    Enlarging the absolute constant $C$ covers this case.
\end{proof}

\end{document}